\documentclass[11pt]{article}

\usepackage{setspace}
\usepackage{natbib}
\usepackage{graphicx}
\usepackage{lscape}
\usepackage{pdfpages}
\usepackage{float}
\usepackage{adjustbox}

\usepackage{arydshln}

\usepackage{arydshln}
\usepackage{hyperref}

\hypersetup{
    colorlinks=true,
    linkcolor=red,
    filecolor=red,      
    urlcolor=red,
    citecolor=blue
}

\usepackage{tikz}
\usepackage{lipsum}
\usepackage{pgffor}

\newtheorem{assumption}{Assumption}

\newtheorem{innercustomgeneric}{\customgenericname}
\providecommand{\customgenericname}{}
\newcommand{\newcustomtheorem}[2]{%
  \newenvironment{#1}[1]
  {%
   \renewcommand\customgenericname{#2}%
   \renewcommand\theinnercustomgeneric{##1}%
   \innercustomgeneric
  }
  {\endinnercustomgeneric}
}

\newcustomtheorem{proposition_b}{Proposition}
\newcustomtheorem{assumption_b}{Assumption}

\usepackage[top=1in, bottom=1in, left=1in, right=1in]{geometry}

\usepackage{amsfonts}%
\usepackage{amsmath}%
\usepackage{amssymb}%
\usepackage{graphicx}
\providecommand{\U}[1]{\protect\rule{.1in}{.1in}}
\newtheorem{theorem}{Theorem}

\newtheorem{proposition}{Proposition}
\newtheorem{remark}[theorem]{Remark}

\newenvironment{proof}[1][Proof]{\noindent \textbf{#1.} }{\  \rule{0.5em}{0.5em}}

\usepackage{footmisc}
\usepackage{caption}

\newcommand\fnote[1]{\captionsetup{font=small}\caption*{#1}}

\begin{document}

\onehalfspacing

\newsavebox{\tablebox}
\newlength{\tableboxwidth}

\title{ 
\huge Synthetic Controls with Imperfect Pre-Treatment Fit\footnote{ We would like to thank Alberto Abadie, Josh Angrist, Ivan Canay, Marcelo Fernandes, Jin Hahn, Guido Imbens, Aureo de Paula, Tong Li, Victor Filipe Martins-da-Rocha, Ricardo Masini, Pedro Sant'Anna, Rodrigo Soares and conference and seminar participants at UC Berkeley, UC Santa Barbara, USC, UDELAR, UCL, University of Bristol,  USP, George Washington University, University of Miami, PUC-Chile, IPEA, PUC-Rio, California Econometrics Conference 2016,  Bristol Econometric Study Group 2016, IAAE conference 2017,  EEA conference 2017, LACEA 2017, and the MIT Conference on Synthetic Controls and Related Methods for comments and suggestions. We also thank  Deivis Angeli, Lucas Barros, and Luis Alvarez for excellent research assistance. Bruno Ferman gratefully acknowledges financial support from FAPESP. An earlier version of this paper circulated under the title ``Revisiting the Synthetic Control Estimator''.}}

\author{
Bruno Ferman\footnote{E-mail: bruno.ferman@fgv.br. Address: Sao Paulo School of Economics, FGV,
Rua Itapeva no. 474, Sao Paulo - Brazil, 01332-000.}    ~~ Cristine Pinto\footnote{E-mail:  cristine.pinto@fgv.br. Address: Sao Paulo School of Economics, FGV,
Rua Itapeva no. 474, Sao Paulo - Brazil, 01332-000.} \\
\\
Sao Paulo School of Economics - FGV \\
\footnotesize
First Draft: June, 2016 \\
\footnotesize
This Draft: January, 2021
}

\date{}
\maketitle

\setstretch{1}

\begin{center}


\

%

\

\textbf{Abstract}

\end{center}

\onehalfspacing

We analyze the properties of the Synthetic Control (SC) and related estimators when the pre-treatment fit is imperfect. In this framework, we show that these  estimators are generally biased if treatment assignment is correlated with unobserved confounders, even when the number of pre-treatment periods goes to infinity. Still, we  show that a demeaned version of the SC method can substantially improve in terms of bias and variance relative to the difference-in-difference estimator. We also derive a specification test for the demeaned SC estimator in this setting with imperfect pre-treatment fit.  Given our theoretical results, we provide practical guidance for applied researchers on how to justify the use of such estimators in empirical applications.

\

\textbf{Keywords:}  synthetic control; difference-in-differences; policy evaluation; linear factor model

\textbf{JEL Codes:} C13; C21; C23

\newpage


\onehalfspacing

\section{Introduction}
\label{Introduction}

In a series of influential papers, \cite{Abadie2003}, \cite{Abadie2010}, and \cite{Abadie2015} proposed the Synthetic Control (SC) method as an alternative to estimate treatment effects in comparative case studies when there is only one treated unit. The main idea of the SC method is to use the pre-treatment periods to estimate weights such that a weighted average of the control units reconstructs the pre-treatment outcomes of the treated unit, and then use these weights to compute the counterfactual of the treated unit in case  it were not treated. 
According to \cite{Athey_Imbens}, \textit{``the simplicity of the idea, and the obvious improvement over the standard methods, have made this a widely used method in the short period of time since its inception''}, making it ``\emph{arguably the most important innovation in the policy evaluation literature in the last 15 years}''.   As one of the main advantages that helped popularize the method, \cite{Abadie2010} derive conditions under which the SC estimator would allow confounding unobserved characteristics with time-varying effects, as long as there exist  weights such that a weighted average of the control units fits the outcomes of the treated unit for  a long set of pre-intervention periods.

In this paper, we analyze the properties of the SC and related estimators when potential outcomes are determined by a linear factor model, which is the structure considered by \cite{Abadie2010} and \cite{Abadie2020} to derive the main theoretical justifications for the SC estimator. Differently from \cite{Abadie2010}, we consider the case in which the pre-treatment fit is imperfect.\footnote{We refer to ``imperfect pre-treatment fit'' as a setting in which it is not assumed existence of weights such that a weighted average of the outcomes of the control unit perfectly fits the outcome of the treated unit for all pre-treatment periods. The perfect pre-treatment fit condition is presented in equation 2 of \cite{Abadie2010}. }    In a model with ``non-diverging'' common factors and a fixed number of control units ($J$), we show that the SC weights  converge in probability  to weights that do \textit{not}, in general, reconstruct the factor loadings of the treated unit  when the number of pre-treatment periods ($T_0$) goes to infinity.\footnote{We refer to  ``non-diverging'' common factors when the pre-treatment average of  of the first and second moments of the common factors converge in probability to a constant. We focus on the SC specification that uses the outcomes of all pre-treatment periods as  predictors. Specifications that use the average of the pre-treatment periods outcomes and other covariates as  predictors are also considered  in Appendix \ref{A_alternatives}. } 
This happens because, in this setting, the SC weights converge to weights that simultaneously attempt to match the factor loadings of the treated unit \textit{and} to minimize the variance of a linear combination of the idiosyncratic shocks. Therefore, weights that reconstruct the factor loadings of the treated unit are not generally the solution to this problem, even if such weights exist. While in many  applications $T_0$ may not be large enough to justify large-$T_0$ asymptotics (e.g.  \cite{Doudchenko}),   our results can also be interpreted as the SC weights not converging to weights that reconstruct the factor loadings of the treated unit \emph{even when $T_0$ is large}. 

 As a consequence, the SC estimator is biased if treatment assignment is correlated with the unobserved heterogeneity in this setting, even when the number of pre-treatment periods goes to infinity. The intuition is the following: if treatment assignment is correlated with  common factors in the post-treatment periods, then we would need a SC unit that is affected in exactly the same way by these common factors as the treated unit, but did not receive the treatment, to obtain an unbiased estimator. However, this condition is not attained when the pre-treatment fit is imperfect, even when $T_0$ is large.\footnote{ \cite{Rothstein} derive finite-sample bounds on the bias of the SC estimator, and show that the  bounds they derive do not converge to zero when $J$ is fixed and $T_0 \rightarrow \infty$. This is consistent with our results, but does not directly imply that the SC estimator is asymptotically biased when $J$ is fixed and $T_0 \rightarrow \infty$. In contrast, our result on the asymptotic bias of the SC estimator imply that it would be impossible to derive bounds that converge to zero in this case. Moreover, we show the conditions under which the estimator is asymptotically biased.  }  Our results  are  not as conflicting with the results from \cite{Abadie2010} as it might appear at first glance. The asymptotic bias of the SC estimator, in our framework,  goes to zero when the variance of the idiosyncratic shocks is small. This is the case in which one should expect to have a close-to-perfect pre-treatment fit  when $T_0$ is large, which is the setting the SC estimator was originally designed for. Our theory complements the theory developed by \cite{Abadie2010}, by considering the properties of the SC estimator when the pre-treatment fit is imperfect.

One important implication of the SC restriction to convex combinations of the control units is that the SC estimator may  be biased even if treatment assignment is only correlated with time-invariant unobserved variables, which is essentially the identification assumption of the difference-in-differences (DID) estimator. We therefore consider a modified SC estimator, where we demean the data using information from the pre-intervention period, and then construct the SC estimator using the demeaned data.\footnote{Demeaning the data before applying the SC estimator is equivalent to relaxing the non-intercept constraint, as suggested, in parallel to our paper, by \cite{Doudchenko}. We formally analyze the implication of this modification to the bias of the SC estimator. The estimator proposed by \cite{Hsiao} relaxes not only the the non-intercept but also the adding-up and non-negativity constraints. We consider the properties of the estimator proposed by \cite{Hsiao}  in Remark \ref{Remark_other}. } An advantage of demeaning is that it is possible to, under some conditions, show that the SC estimator dominates the  DID estimator in terms of variance and bias in this setting. {Moreover, we provide a specification test for the validity of the demeaned SC estimator in this setting with an imperfect pre-treatment fit. Finally, we also show that, in a setting with both non-diverging and diverging common factors, diverging common shocks would not generate asymptotic bias in the demeaned SC estimator, but we need that  treatment assignment is uncorrelated with the non-diverging common factors to guarantee  asymptotic unbiasedness.\footnote{For this result, we need an assumption of existence of weights that reconstruct the factor loadings of the treated unit associated with the diverging common factors. This result holds for the demeaned SC estimator, but not  for the original SC estimator. }       }

If potential outcomes follow a linear factor model structure, then it would be possible to construct a counterfactual for the treated unit if we could consistently estimate the factor loadings.\footnote{Assuming that it is possible to construct a linear combination of the factor loadings of the control units that reconstructs the factor loadings of the treated unit, then this linear combination of the control units' outcomes would provide an unbiased counterfactual for the treated unit.  } However, with fixed $J$, it is only possible to estimate factor loadings consistently under strong assumptions on the idiosyncratic shocks (e.g., \cite{Bai2003} and \cite{anderson1984introduction}). Therefore, the asymptotic bias we find for the SC  estimator is consistent with the results from a large literature on factor models. We show that the asymptotic bias we derive for the SC estimator also applies to other related panel data approaches that have been studied in the context of an imperfect pre-treatment fit, such as \cite{Hsiao},  \cite{Li}, \cite{Carvalho2015},  \cite{Carvalho2016b},  and \cite{Masini}, in settings with fixed $J$. We show that these papers rely on assumptions that implicitly imply no selection on unobservables, which clarifies why their consistency/unbiasedness results when $J$ is fixed are not conflicting with our main results. 

 Also consistent with the literature on factor models, if we impose restrictions on the idiosyncratic shocks, then there are asymptotically unbiased alternatives. For example,   \cite{Robust_SC} propose a  de-noising algorithm, but it relies on idiosyncratic errors being serially uncorrelated.\footnote{This is also the case for an  IV-like SC estimator we presented in an earlier version of this paper \citep{FP_old}.} However, this may not be an appealing assumption  in common applications. To the best of our knowledge, there is no estimator that is asymptotically valid in settings with fixed $J$ without assuming such kind of additional assumptions.  Finally, \cite{Powell2} proposes a 2-step estimation in a setting with fixed $J$ in which the SC unit is constructed based on the fitted values of the outcomes on unit-specific time trends. However,  we show that the demeaned  SC method is already very efficient in controlling for polynomial time trends

 When both $J$ and $T_0$ diverge, \cite{Magnac},  \cite{XU},   \cite{Imbens_matrix}, and \cite{SDID} provide alternative estimation methods that are asymptotically valid when the number of both pre-treatment periods and controls increase. This is also consistent with the literature on linear factor models, which shows that these models can be consistently estimated in large panels (e.g., \cite{Bai2003}, \cite{baing}, \cite{Bai}, and \cite{Martin}).  \cite{Ferman} provides conditions under which the original and the demeaned SC estimators are also asymptotically unbiased in this setting with large $J$/large $T_0$. The  main requirement is that, as the number of control units increases, there are weights diluted among an increasing number of control units that recover the factor loadings of the treated unit. In contrast, our results on the bias of the SC estimator provide a better approximation for the properties of the SC estimator for cases in which this condition on the weights is not valid, and/or when $J$ and $T_0$ are roughly of the same size, but they are not large enough, so that a large $T_0$/large $J$ asymptotics does not provide a good approximation.

The remainder of this paper proceeds as follows. In Section \ref{SC_model} we describe our setting and provide a  brief review of the SC estimator. The main results are presented in Section \ref{Sec_main_results}. We then present a Monte Carlo (MC) simulation in Section \ref{Sec_MC}, and an empirical illustration in Section \ref{Sec_EI}.  In Section \ref{recommendations} we provide a guideline for applied researchers on how to justify the use of the SC method, based on our theoretical results.  We conclude in Section \ref{Conclusion}.

\section{Base Model} \label{SC_model}

Suppose we have a balanced panel of  $J+1$ units indexed by $j = 0,...,J$ observed on a total of $T$ periods.  We want to estimate the treatment effect of a policy change that affected only unit $j=0$, and we have information before and after the policy change.  Let $\mathcal{T}_0$ ($\mathcal{T}_1$) be the set of time indices in the pre-treatment (post-treatment) periods. We assume that potential outcomes follow a linear factor model. 

\begin{assumption}[potential outcomes] \label{assumption_LFM}
\normalfont
Potential outcomes when unit $j$ at time $t$ is treated ($y_{jt}^I$) and non-treated ($y_{jt}^N$) are given by

\begin{eqnarray} \label{model}
\begin{cases} y_{jt}^N = c_j + \delta_t + \lambda_t \mu_j + \epsilon_{jt}  \\ 
y_{jt}^I = \alpha_{jt} + y_{jt}^N, \end{cases}
\end{eqnarray}
where $\delta_t$ is an unknown common factor with constant factor loadings across units, $c_j$ is an unknown time-invariant fixed effect, $\lambda_t$ is a $(1 \times F)$ vector of common factors, $\mu_j$ is a $(F \times 1)$ vector of unknown factor loadings, and the error terms $\epsilon_{jt}$ are unobserved idiosyncratic shocks.

\end{assumption}

In principle, the terms $\delta_t$ and $c_j$  could be included in the linear factor structure  $\lambda_t\mu_j$. We include these separately because we want to consider $\lambda_t$ as a vector of common factors that do not have constant effects across units and that do not include a time-invariant fixed effect. Therefore, we can think of   $\lambda_t$ as time-varying unobservables that may affect different units differently. In order to simplify the exposition of our main results, we consider the model without  observed covariates $Z_j$. In Appendix Section \ref{theta} we consider the model with covariates. 

The treatment effect on unit $j$ at time $t$ is given by $\alpha_{jt}$, and the main goal of the SC method is to estimate the effect of the treatment for unit 0 for each post-treatment  $t$, that is $\{ \alpha_{1t} \}_{t \in \mathcal{T}_1}$. However, we only observe $y_{jt} = d_{jt} y_{jt}^I  +  (1-d_{jt}) y_{jt}^N$, where $d_{jt}=1$ if unit $j$ is treated at time $t$. 

We treat the vector of unknown factor loadings ($\mu_j$) and the treatment assignment as fixed, and consider the properties of the SC estimator under a repeated sampling framework over the distributions of the common factors ($\lambda_t$) and of the idiosyncratic shocks ($\epsilon_{jt}$). Alternatively, we can think that we have an underlying model where treatment assignment, $\mu_j$, $\lambda_t$, and $\epsilon_{jt}$ are stochastic, but we are conditioning on the treatment assignment and on the factor loadings.
Assumption \ref{assumption_sample} defines the observed sample.

\begin{assumption}[sampling] \label{assumption_sample}

\normalfont We observe a realization of $\{ y_{0t} ,..., y_{Jt} \}_{t \in \mathcal{T}_0 \cup \mathcal{T}_1}$, where  $y_{jt} = d_{jt} y_{jt}^I  +  (1-d_{jt}) y_{jt}^N$, while $d_{jt}=1$ if $j=0$ and $t\in \mathcal{T}_1$, and zero otherwise. Potential outcomes are determined by equation (\ref{model}). We treat $\{c_j, \mu_j\}_{j=0}^J$ as fixed, and $\{ \lambda_t \}_{t \in \mathcal{T}_0 \cup \mathcal{T}_1}$ and $\{ \epsilon_{jt} \}_{t \in \mathcal{T}_0 \cup \mathcal{T}_1}$ for $j=0,...,J$ as stochastic.

\end{assumption}

 In the assumption below we consider the identification assumption usually considered in the SC literature. 

\begin{assumption}[idiosyncratic shocks] \label{assumption_exogeneity}

\normalfont $\mathbb{E}[\epsilon_{jt} ]=0$  for all $j \in \{0,1,...,J\}$ and $t \in \mathcal{T}_1 \cup \mathcal{T}_0$.

\end{assumption}

Assumption \ref{assumption_exogeneity}, combined with the fact that we consider treatment assignment and factor loadings as fixed, compose the main restrictions we impose on the treatment assignment mechanism. It is easier to think about the assignment mechanism if we consider an underlying model in which treatment assignment and factor loadings are stochastic, and the expectation in Assumption \ref{assumption_exogeneity} is conditional on the realization of these variables. In this case, Assumption \ref{assumption_exogeneity} implies that idiosyncratic shocks are mean-independent from the treatment assignment. However, it does not impose any restriction on the dependence between treatment assignment and the factor structure. In particular,  Assumption \ref{assumption_exogeneity} does  not impose any restriction on the distribution of $\lambda_t$ for $t \in \mathcal{T}_1$. We refer to that as ``selection on unobservables'', meaning that treatment assignment may be correlated with the factor structure, but is uncorrelated with the idiosyncratic shocks.\footnote{This assumptions is essentially the same as the ones considered by, for example, \cite{Abadie2010}, \cite{Magnac} and  \cite{Rothstein}  (in their Section 4.1), where they assume unconfoundness conditional on the unobserved factor loadings.}

Let $\boldsymbol{\mu} \equiv [\mu_1 \hdots \mu_J]'$, $\mathbf{c} \equiv [c_1 \hdots c_J]'$,   $\mathbf{y}_t \equiv (y_{1t}, \hdots, y_{Jt})$ and $\boldsymbol{\epsilon}_t \equiv (\epsilon_{1t}, \hdots, \epsilon_{Jt})$. Following the original SC papers, we start restricting to convex combinations of the control units, so we consider weights in  $\Delta^{J-1} \equiv \{ (w_1,...,w_J) \in \mathbb{R}^{J} | w_j \geq 0 \mbox{ and } \sum_{j=1}^J w_j = 1\}$. We define  $\widetilde \Phi = \{  \textbf{w} \in \Delta^{J-1} ~ | ~ \mu_0 = \boldsymbol{\mu}' \mathbf{w} \mbox{ and } c_0 = \mathbf{c}' \mathbf{w} \}$. Therefore, $\mathbf{w} \in \widetilde \Phi$ is such that a weighted average of the control units absorbs all time correlated shocks of unit 0, $  \lambda_t \mu_0$, and also absorbs the time-invariant fixed effects.  Assuming $\widetilde \Phi \neq \varnothing$, if we knew $\textbf{w}^\ast \in  \widetilde \Phi $, then we could consider an \emph{infeasible} SC estimator using these weights, $\hat \alpha_{0t}^\ast =  y_{0t} -  \mathbf{y}_t '  {\mathbf{w}^\ast}   $. For a given $t \in \mathcal{T}_1$, we would have 
\begin{eqnarray} \label{infeasible}
\hat \alpha^\ast_{0t} = y_{0t} - \mathbf{y}_t ' \mathbf{w}^\ast    =  \alpha_{0t} + \left( \epsilon_{0t} - \boldsymbol{\epsilon}_t ' \mathbf{w}^\ast \right).
\end{eqnarray}

Therefore, under Assumption \ref{assumption_exogeneity},  $\mathbb{E}[\hat \alpha^\ast_{0t} ] = \alpha_{0t}$, which implies that this infeasible SC estimator is unbiased.    Intuitively, the infeasible SC estimator constructs a SC unit for the counterfactual of $y_{0t}$ that is affected in the same way as unit 0 by each of the common factors (that is, $\mu_0 = \boldsymbol{\mu} ' \mathbf{w}^\ast$) and has the same time-invariant fixed effect ($c_0 = \mathbf{c} ' \mathbf{w}^\ast$), but did not receive treatment. Therefore, the only difference between unit 0 and this SC unit, beyond the treatment effect, would be given by the idiosyncratic shocks, which are assumed  to have mean zero  (Assumption \ref{assumption_exogeneity}), implying that this infeasible SC estimator is  unbiased.

It is important to note that  \cite{Abadie2010} do note make any assumption on $\widetilde \Phi \neq \varnothing$. Instead, they consider that there is a set of weights $\widetilde{\mathbf{w}}^\ast \in \Delta^{J-1}$ that satisfies $y_{0t} = \mathbf{y}_t'  \widetilde{\mathbf{w}}^\ast$ for all $t \in \mathcal{T}_0$.\footnote{\cite{Abadie2010} assume that such weights also provide perfect balance in terms of observed covariates. \cite{FB} analyze the case in which the perfect balance on covariates assumption is dropped, but there is still perfect balance on pre-treatment outcomes.}  We call the existence of such weights $\widetilde{\mathbf{w}}^\ast \in \Delta^{J-1}$ as a ``perfect pre-treatment fit'' condition. While subtle, this reflects a crucial difference between our setting and the setting considered in the original SC papers. \cite{Abadie2010} and \cite{Abadie2015}    consider the properties of the SC estimator  conditional on having a perfect pre-intervention fit. As stated by \cite{Abadie2015}, they \textit{``do not recommend using this method when the pretreatment fit is poor or the number of pretreatment periods is small''}. 

\cite{Abadie2010} provide conditions under which existence of $\widetilde{\mathbf{w}}^\ast \in \Delta^{J-1}$ such that $y_{0t} = \mathbf{y}_t'  \widetilde{\mathbf{w}}^\ast$ for all $t \in \mathcal{T}_0$  (for large $T_0$) implies that $\mu_{0} \approx \boldsymbol{\mu} ' \widetilde{\mathbf{w}}^\ast$.  In this case, the bias of the SC estimator would be bounded by a function that goes to zero when $T_0$ increases. We depart from the original SC setting in that we consider a setting with imperfect pre-treatment fit, meaning that  we do not assume existence of  $\widetilde{\mathbf{w}}^\ast \in \Delta^{J-1}$ such that $y_{0t} = \mathbf{y}_t'  \widetilde{\mathbf{w}}^\ast$ for all $t \in \mathcal{T}_0$. The motivation to analyze the SC method in our setting is that the SC estimator has been widely used even when the pre-treatment fit is poor. Therefore, it is important to understand the properties of the estimator in this setting. Moreover, we show that the estimator can provide important improvements relative to DID even when the fit is imperfect, although in this case we should be more careful about the conditions for unbiasedness. 

In order to implement their method, \cite{Abadie2010} recommend a nested minimization problem using the pre-intervention data to estimate the SC weights.  We focus on the case where one includes all pre-intervention outcome values as  predictors. In this case, the nested optimization problem proposed by \cite{Abadie2010}  simplifies to\footnote{ See \cite{Kaul2015} and \cite{Doudchenko}. }
\begin{eqnarray}  \label{objective_function}
\widehat{\textbf{w}}^{\mbox{\tiny SC}} &=& \underset{{\textbf{w} \in \Delta^{J-1}}}{\mbox{argmin}} \frac{1}{T_0} \sum_{t \in \mathcal{T}_0} \left[ y_{0t} - \mathbf{y}_{t}' \mathbf{w}   \right]^2.
\end{eqnarray}

For a given $t \in \mathcal{T}_1$, the SC estimator is then defined by $\hat \alpha_{0t} = y_{0t} - \mathbf{y}_t ' \widehat{\mathbf{w}}^{\mbox{\tiny SC}}$.
\cite{FPP} provide conditions under which the SC estimator using all pre-treatment outcomes as predictors will be asymptotically equivalent, when $T_0 \rightarrow \infty$, to any alternative SC estimator such that the number of pre-treatment outcomes used as predictors goes to infinity with $T_0$, even for specifications that include other covariates. Therefore, our results are also valid for these SC specifications under these conditions. In Appendix \ref{A_alternatives} we also consider  SC estimators using (1)  the average of the pre-intervention outcomes as predictor,  and (2)  other time-invariant covariates in addition to the average of the pre-intervention outcomes as predictors.

\section{Main results } \label{Sec_main_results}

We consider the asymptotic properties of the SC and alternative estimators when $T_0 \rightarrow \infty$ and $J$ is fixed. As we discuss in Remark \ref{R_small_T}, our results are also relevant for the case in which $T_0$ is small. We consider the properties of the original SC estimator in Section \ref{Bias_SC} in a setting in which common factors are ``non-diverging'', in the sense that the second moments of the pre-treatment averages of the common factors and of the idiosyncratic shocks converge in probability to non-stochastic constants. We propose and analyze a demeaned version of the SC estimator in Section \ref{Sub_demeaned} in this setting. Then we discuss in Section \ref{Setting2} a setting in which some common factors are ``diverging'', in the sense that the variance of the idiosyncratic shocks become asymptotically irrelevant relative to the variance of the common shocks when $T_0$ increases. This setting is analyzed in more details in Appendix \ref{Appendix_diverging}.

\subsection{Asymptotic bias of the original SC estimator} \label{Bias_SC}

We consider a settings in which the  pre-treatment averages of the first and second moments of the common factors and the idiosyncratic shocks converge in probability to non-stochastic constants. {Importantly, note we do not require that the  observed outcomes $y_{jt}$  satisfy these conditions,  because we do not impose any restriction on $\delta_t$. We discuss in Section  \ref{Setting2} the case in which diverging common shocks may have heterogeneous effects across units.} Let $\varepsilon_t = (\epsilon_{0t},...,\epsilon_{Jt})$.

\begin{assumption}[common and idiosyncratic shocks] \label{assumptions_lambda}

\normalfont  $\frac{1}{T_0} \sum_{t \in \mathcal{T}_0} \lambda_t    \buildrel p \over \rightarrow  0$, $\frac{1}{T_0} \sum_{t \in \mathcal{T}_0} \varepsilon_{t}    \buildrel p \over \rightarrow   0$, \\ $\frac{1}{T_0} \sum_{t \in \mathcal{T}_0} \lambda_t'\lambda_t    \buildrel p \over \rightarrow   \Omega_0$ positive semi-definite, $\frac{1}{T_0} \sum_{t \in \mathcal{T}_0} \varepsilon_{t}\varepsilon_{t}'    \buildrel p \over \rightarrow   \sigma^2_\epsilon I_{J+1}$, and   $\frac{1}{T_0} \sum_{t \in \mathcal{T}_0} \varepsilon_{t}\lambda_{t}    \buildrel p \over \rightarrow  0$ when $T_0 \rightarrow \infty$. 

\end{assumption}

Assumption \ref{assumptions_lambda} allows for serial correlation for both idiosyncratic shocks and common factors. The only restriction on the serial correlation is that we can apply a law of large numbers so that these pre-treatment averages converge in probability.   We assume $\frac{1}{T_0} \sum_{t \in \mathcal{T}_0} \varepsilon_{t}\varepsilon_{t}'    \buildrel p \over \rightarrow   \sigma^2_\epsilon I_{J+1}$ in order to simplify the exposition of our results. However, this  can be easily replaced by $\frac{1}{T_0} \sum_{t \in \mathcal{T}_0} \varepsilon_{t}\varepsilon_{t}'    \buildrel p \over \rightarrow  \Sigma$ for any symmetric positive definite $(J+1) \times (J+1)$ matrix $\Sigma$, so that idiosyncratic shocks may be heteroskedastic and correlated across $j$. Assuming that $\frac{1}{T_0} \sum_{t \in \mathcal{T}_0} \lambda_t    \buildrel p \over \rightarrow  \omega_0$, setting $\omega_0 = 0 $ is without loss of generality.\footnote{If $\omega_0 \neq 0$, then we can consider an observably equivalent model with $\omega_0 = 0$ by adjusting $c_j$. }
Assumption \ref{assumptions_lambda} would be satisfied if, for example,   $(\varepsilon'_{t},\lambda_t)$ is $\alpha-$mixing with exponential speed, with uniformly bounded fourth moments in the pre-treatment period, and  $\varepsilon_{t}$ and $\lambda_t$ are independent. Note that this would allow the distribution of $\lambda_t$ to be different when we consider pre-treatment periods closer to the assignment of the treatment.   In this case,  $\lambda_t$ would not be stationary, but Assumption \ref{assumptions_lambda} would still hold. Finally, note that we do not impose any restriction on $\delta_t$.

We consider in Proposition \ref{main_result} the asymptotic distributions of the original SC in this setting. 

\begin{proposition} \label{main_result}
\normalfont Under Assumptions \ref{assumption_LFM} to  \ref{assumptions_lambda},  { $\mathbf{\widehat w}^{\mbox{\tiny SC}}  \buildrel p \over \rightarrow  \mathbf{\bar w}$} when $T_0 \rightarrow \infty$, where  $(c_0,\mu_0) \neq (\mathbf{c}'\mathbf{ \bar w},\boldsymbol{\mu}'\mathbf{ \bar w})$, unless $ \sigma_\epsilon^2=0 $ or $\exists \textbf{w} \in \widetilde \Phi |  \textbf{w} \in  \underset{{\textbf{w} \in \Delta^{J-1}}}{\mbox{argmin}} \left\{  \mathbf{w}'\mathbf{w} \right\}$. Moreover, for $t \in \mathcal{T}_1$,
\begin{eqnarray} \label{SC_asymptotic_distribution}
\hat \alpha_{0t} = y_{0t} -  \mathbf{y}_t ' \widehat{\mathbf{w}}^{\mbox{\tiny SC}}     \buildrel p \over \rightarrow  \alpha_{0t} +  \lambda_t \left(\mu_0 - \boldsymbol{\mu}' \mathbf{\bar w} \right) + \left(c_0 - \mathbf{c}' \mathbf{\bar w} \right)  + \left( \epsilon_{0t} - \boldsymbol{\epsilon}_t ' \mathbf{\bar w} \right)  \mbox{ when } T_0 \rightarrow \infty.
\end{eqnarray}
\end{proposition}

Proposition \ref{main_result} shows that the weights of the original  SC estimators will generally not converge to weights that recover the factor loadings of the treated unit. The intuition is that  $\mathbf{\widehat w}^{\mbox{\tiny SC}} $ converges in probability to $\mathbf{w} \in \Delta^{J-1}$  that minimizes the probability limit of equation (\ref{objective_function}), which is given by
\begin{eqnarray} \label{Q0}
Q_0(\mathbf{w}) = \left[ \left( c_0 - \mathbf{c}'\mathbf{w}   \right)^2 + \left( \mu_0 - \boldsymbol{\mu}'\mathbf{w}   \right)' \Omega_0  \left( \mu_0 - \boldsymbol{\mu}'\mathbf{w}  \right) \right]+ \sigma_\epsilon^2 \left( 1+ \mathbf{w}' \mathbf{w}    \right).
\end{eqnarray}

This objective function has two parts. The first one reflects the presence of  common factors $\lambda_t$ and differences in the fixed effects that remain after we choose the weights to construct the SC unit. If $\widetilde \Phi \neq \varnothing$, then we can set this part equal to zero by choosing $\textbf{w}^\ast \in \widetilde \Phi$. However, this objective function also depends on the variance of a weighted average of the idiosyncratic shocks $\epsilon_{jt}$, implying that choosing  $\textbf{w}^\ast \in \widetilde \Phi$ will not generally be the solution to this problem. As a consequence, the SC weights will generally converge to weights that do not recover the factor loadings of the treated unit, even if $\widetilde \Phi \neq \varnothing$. The SC weights would only asymptotically recover the factor loadings of the treated unit if  $ \sigma_\epsilon^2=0 $ or $\exists \textbf{w} \in \widetilde \Phi   \mbox{ such that } \textbf{w} \in \mbox{argmin}_{\textbf{w} \in \Delta^{J-1}} \left\{ \mathbf{w}'\mathbf{w}   \right\}$. {Given this rationale, such distortion on the SC weights will tend to be smaller when the common trends are much stronger than the idiosyncratic shocks (so that the second part of the objective function $Q_0(\mathbf{w})$ becomes less relevant).} We present details of proof in Appendix \ref{Prop1}. Another intuition for this result is that the outcomes of the controls work as proxy variables for the factor loadings of the treated unit, but they are measured with error. We present this interpretation in more detail in Appendix \ref{finite_T}.

Proposition \ref{main_result} also shows that the SC estimator  converges in probability  to the parameter we want to estimate ($ \alpha_{0t}$) plus linear combinations of contemporaneous idiosyncratic shocks and  common factors.\footnote{For simplicity, we consider the case in which $\alpha_{0t}$ is a fixed parameter. More generally, we could consider $\alpha_{0t}$  stochastic, and re-define the parameter of interest as $ \mathbb{E}[\alpha_{0t}] $. The intuition for all results would remain unchanged.} By Assumption \ref{assumption_exogeneity}, $\mathbb{E}[\epsilon_{jt}]=0$,  so whether this estimator is asymptotically unbiased depends crucially on the differences in how the treated and the SC units are affected by the common shocks, $\lambda_t \left(\mu_0 - \boldsymbol{\mu}' \mathbf{\bar w} \right)$, and on whether the SC unit reconstructs  $c_0$. We can guarantee asymptotic unbiasedness for the original SC estimator if we assume that, for $t \in \mathcal{T}_1$, $\mathbb{E} \left[ \lambda^k_t  \right]=0$ for all common factors $k$ such that $\mu^k_0 \neq \sum_{j \neq 0} \bar w_j^{\mbox{\tiny SC}} \mu^k_j$, and $c_0 = \mathbf{c}'\mathbf{w} $.\footnote{There could also be linear combinations of biases arriving from different common factors that end up cancelling out, but we see that as uninteresting  ``knife-edge'' cases.}  Therefore, once we relax the perfect fit condition, Assumption \ref{assumption_exogeneity}, which allows for selection on unobservables, is not sufficient to guarantee that the SC estimator is asymptotically unbiased. More specifically, Proposition \ref{main_result} shows that, once we relax the perfect fit condition, the original SC estimator will generally be asymptotically biased when treatment assignment is correlated with time-varying unobservables, and when the SC weights fail to recover the levels of the treated unit. This second conclusion implies that the SC estimator may be biased even when a DID estimator would be unbiased.

\begin{remark}
\normalfont
The discrepancy of our results with the results from \cite{Abadie2010} arises because we consider different frameworks. \cite{Abadie2010}  consider the properties of the SC estimator conditional on having a perfect pre-treatment fit. Our results are  not as conflicting with the results from \cite{Abadie2010} as they may appear at first glance. In a model with non-diverging common factors, the probability that one has a dataset at hand such that the SC weights provide a close-to-perfect pre-intervention fit with a moderate $T_0$ is close to zero, unless the variance of the idiosyncratic shocks is small. Therefore, our results agree with the theoretical results from  \cite{Abadie2010} in that the asymptotic bias of the SC estimator should be small in situations where one would expect to have a close-to-perfect fit for a large $T_0$. 

\end{remark}

\begin{remark}
\normalfont \label{R_small_T}
While many SC applications do not have a large  number of pre-treatment periods to justify large-$T_0$ asymptotics (see, for example,  \cite{Doudchenko}),   our results can also  be interpreted as the SC weights not converging to weights that reconstruct the factor loadings of the treated unit when $J$ is fixed \emph{even when $T_0$ is large}. In Appendix \ref{finite_T}, we show that the problem we present remains if we consider a setting with finite $T_0$.

\end{remark}

\begin{remark}
\normalfont

Related to Remark \ref{R_small_T}, the bias we derive for the SC estimator does not come from the difficulty in having a perfect pre-treatment fit when $T_0$ is large and $J$ is fixed. Rather, this would remain a problem when $J$ is fixed, even if $T_0$ is small enough so that a perfect pre-treatment fit is achieved due to over-fitting. 

\end{remark}

\begin{remark}
\normalfont

Having a larger $T_0$ could make the problem worse if the linear factor model for the post-treatment periods does not approximate well the potential outcomes in earlier periods. For example, imagine we have $y^N_{jt} = \lambda_t \mu_j + \epsilon_{jt}$ for the post-treatment periods and for the $M$ periods before the treatment, and $y^N_{jt} = \lambda_t \tilde \mu_j + \epsilon_{jt}$ for earlier periods, with $\mu_j$ potentially different from $\tilde \mu_j$. In this case, including more pre-treatment periods would likely induce more bias in the SC estimator, because in this case the first part of the objective function in (\ref{Q0}) would be minimized with a $\widetilde{\mathbf{w}}$ such that $ \mu_0 \neq  \boldsymbol{\mu}'\widetilde{\mathbf{w}}$. In any case, our results from Proposition \ref{main_result} should be interpreted as the SC estimator generally being asymptotically biased if there is selection on unobservables \emph{even when} $T_0 \rightarrow \infty$ \emph{and} the factor loadings are stable for all periods when $T_0 \rightarrow \infty$.

\end{remark}

\subsection{Comparison with DID estimator  \& the demeaned SC estimator } \label{Sub_demeaned}

In contrast to the SC estimator, the DID estimator  for the treatment effect in a given post-intervention period $t \in \mathcal{T}_1$ would be given by
\begin{eqnarray} 
\hat \alpha^{\tiny DID}_{0t} &=& y_{0t} - \frac{1}{J} \mathbf{y}_t ' \mathbf{i}  - \frac{1}{T_0}\sum_{\tau \in \mathcal{T}_0} \left[  y_{0\tau} - \frac{1}{J} \mathbf{y}_\tau ' \mathbf{i}      \right],
\end{eqnarray}
where $\mathbf{i}$ is a $J \times 1$ vector of ones.\footnote{Note that the DID estimator in this case with one treated unit is numerically the same as the two-way fixed effects (TWFE) estimator using unit and time fixed effects. Since the goal in the SC literature is to estimate the effect of the treatment for unit 1 at a specific date $t$, this circumvents the problem of aggregating heterogeneous effects, as considered by \cite{chaisemartin2018twoway}, \cite{Pedro}, \cite{Imbens_DID}, and \cite{Bacon} in the DID setting.}   Under  Assumptions \ref{assumption_LFM}, \ref{assumption_sample},  and \ref{assumptions_lambda}, we have that
\begin{eqnarray} 
\hat \alpha^{\tiny DID}_{0t} &\buildrel p \over \rightarrow  &  \alpha_{0t} + \left( \epsilon_{0t}  - \frac{1}{J}  \boldsymbol{\epsilon}_{t} ' \mathbf{i} \right) +  \lambda_t  \left(  \mu_{0}  - \frac{1}{J} \boldsymbol{\mu} ' \mathbf{i}  \right)  \mbox{ when } T_0 \rightarrow \infty.
\end{eqnarray}

Therefore, the DID estimator will be asymptotically unbiased in this setting if $\mathbb{E}[\lambda_t ] = 0$ for the factors such that  $ \mu_{0}  \neq \frac{1}{J} \boldsymbol{\mu} ' \mathbf{i}$. This would be the case if treatment assignment is uncorrelated with the time-varying common factors. Differently from the SC estimator, however, the DID estimator would \emph{not} be biased if the average of the control units do not recover the fixed effect of the treated unit. 

As an  alternative to the standard SC estimator, we suggest a modification in which we calculate the pre-treatment average for all  units and demean the data. This is equivalent to  a generalization of the SC method suggested, in parallel to our paper,  by \cite{Doudchenko},  which includes an intercept parameter in the minimization problem to estimate the SC weights and construct the counterfactual.\footnote{Relaxing the non-intercept constraint was already a feature of \cite{Hsiao}. The difference here is that we relax this constraint while maintaining the adding-up and non-negativity constraints, which allows us to rank the demeaned SC with the DID estimator under some conditions.} Here we formally consider the implications of this alternative on the bias and variance of the SC estimator.  

The demeaned SC estimator is given by $\hat \alpha^{\mbox{\tiny SC$'$}}_{0t} = y_{0t} -  \mathbf{y}_t ' \widehat{\mathbf{w}}^{\mbox{\tiny SC$'$}} -  (\bar y_{0} - \mathbf{\bar y}' \widehat{\mathbf{w}}^{\mbox{\tiny SC$'$}    }  )$, where $\bar y_0$ is the pre-treatment average of unit $0$, and $\mathbf{\bar y}$ is an $J \times 1$ vector with the pre-treatment averages of the controls. We define $\Phi = \{  \textbf{w} \in \Delta^{J-1} ~ | ~ \mu_0 = \boldsymbol{\mu}' \mathbf{w} \}$. In this case, the weights $\widehat{\textbf{w}}^{\mbox{\tiny SC$'$}} $ are given by
\begin{eqnarray}  \label{Q_demeaned}
\widehat{\textbf{w}}^{\mbox{\tiny SC$'$}} =  \underset{{\textbf{w} \in \Delta^{J-1}}}{\mbox{argmin}} \frac{1}{T_0} \sum_{t \in \mathcal{T}_0} \left[ y_{0t} - \mathbf{y}_t'\mathbf{w}  - \left(\bar y_{0} -\mathbf{\bar y} ' \mathbf{w} \right)   \right]^2.   \end{eqnarray}

\begin{proposition} \label{Prop_demeaned} 
\normalfont Under Assumptions \ref{assumption_LFM}, \ref{assumption_sample}, \ref{assumption_exogeneity} and \ref{assumptions_lambda},  { $\widehat{\textbf{w}}^{\mbox{\tiny SC$'$}}   \buildrel p \over \rightarrow  \mathbf{\bar w}^{\mbox{\tiny SC$'$}}$} when $T_0 \rightarrow \infty$, where  $\mu_0 \neq \boldsymbol{\mu} ' \mathbf{\bar w}^{\mbox{\tiny SC$'$}} $, unless $ \sigma_\epsilon^2=0 $ or $\exists \textbf{w} \in \Phi |  \textbf{w} \in  \underset{{\textbf{w} \in \Delta^{J-1}}}{\mbox{argmin}} \left\{  \mathbf{w}'\mathbf{w} \right\}$. Moreover, for $t \in \mathcal{T}_1$,
\begin{eqnarray}
\hat \alpha^{\mbox{\tiny SC$'$}}_{0t}   \buildrel p \over \rightarrow  \alpha_{0t} + \left( \epsilon_{0t} -  \boldsymbol{\epsilon}_t'\mathbf{\bar w}^{\mbox{\tiny SC$'$}} \right) + \lambda_t \left(\mu_0 - \boldsymbol{\mu} ' \mathbf{\bar w}^{\mbox{\tiny SC$'$}} \right)  \mbox{ when } T_0 \rightarrow \infty.
\end{eqnarray}
\end{proposition}

Therefore, both the demeaned SC  and the DID estimators are asymptotically unbiased when $\mathbb{E}[\lambda_t ] = 0$ for $t \in \mathcal{T}_1$.\footnote{This is a sufficient condition. More generally, the demeaned SC estimator would be asymptotically unbiased if $\mathbb{E}[\lambda^k_t] = \omega_0$ for $t \in \mathcal{T}_1$ for any common factor $k$ such that $\mu^k_0 \neq \sum_{j \neq 0} \bar w^{\mbox{\tiny SC$'$}}_j \mu^k_j$. However, as  we show in Proposition \ref{Prop_demeaned}, if $\sigma^2_\epsilon>0$, then we would only have $\mu^k_0 = \sum_{j \neq 0} \bar w^{\mbox{\tiny SC$'$}}_j \mu^k_j$ in knife-edge cases. Therefore, we focus on the sufficient condition $\mathbb{E}[\lambda_t ] = 0$ for $t \in \mathcal{T}_1$.} Differently from the original SC estimator, these estimators do not require that the weights recover the fixed effect of the treated unit for unbiasedness.   The proof is essentially the same as the one for Proposition \ref{main_result} (details in Appendix \ref{Proof_demeaned}).

With additional assumptions on $(\epsilon_{0t},...,\epsilon_{Jt},\lambda_t')$ in the post-treatment periods,   we can also assure that the demeaned SC estimator is asymptotically  more  efficient than DID. 

\begin{assumption}[Stability in the pre- and post-treatment periods]

\label{A5}

\normalfont For $t \in \mathcal{T}_1$, $\mathbb{E}[\lambda_t  ]=    0$, $\mathbb{E}[\epsilon_{t}]= 0$, $\mathbb{E}[\lambda_t'\lambda_t ] = \Omega_0$, and $\mathbb{E}[\epsilon_t\epsilon_t' ] =  \sigma^2_\epsilon I_{J+1}$, $cov(\epsilon_t,\lambda_t)=0$.

\end{assumption}

Assumptions \ref{assumptions_lambda} and \ref{A5} imply that  idiosyncratic shocks and common factors have the same first and second moments  in the pre- and post-treatment periods. Again, the assumptions that idiosyncratic errors are homoskedastic is made just for simplification. What is crucial in this assumption is that the variance/covariance matrix of the idiosyncratic shocks in the post-treatment periods is the same as the long-run variance/covariance matrix of the idiosyncratic shocks in the pre-treatment periods. From Proposition \ref{Prop_demeaned},  Assumption \ref{A5}  implies that the demeaned SC estimator is asymptotically unbiased. We now show that this assumption also implies that the demeaned SC estimator has lower asymptotic MSE than  the DID estimator.

\begin{proposition} \label{Prop_demeaned_efficient}

\normalfont Under Assumptions \ref{assumption_LFM}  to \ref{A5},  the demeaned SC estimator ($\hat \alpha^{\mbox{\tiny SC$'$}}_{0t}$) dominates  the DID estimator ($\hat \alpha^{\tiny DID}_{0t} $) in terms of asymptotic MSE when $T_0 \rightarrow \infty$.
\end{proposition}

The intuition of this result  is that, under Assumption \ref{A5}, the demeaned SC weights converge to weights that minimize a function $\Gamma(\mathbf{w})$ such that $\Gamma(\mathbf{\bar w}^{\mbox{\tiny SC$'$}}) = a.var(\hat \alpha^{\mbox{\tiny SC$'$}}_{0t}  ) $, and  $\Gamma(\{\frac{1}{J},...,\frac{1}{J} \}) = a.var(\hat \alpha^{\mbox{ \tiny DID}}_{1t}  ) $. Therefore, it must be that the asymptotic variance of $\hat \alpha^{\mbox{\tiny SC$'$}}_{0t} $ is weakly lower than the variance of $\hat \alpha^{\mbox{ \tiny DID}}_{1t} $. Moreover, these  estimators are unbiased under these assumptions (details in Appendix \ref{Proof_demeaned_efficient}). 

If we had that   $\mathbb{E}[\lambda_t ] \neq 0$ for $t \in \mathcal{T}_1$, then both the demeaned SC  and the DID estimators would generally be asymptotically biased.  In general, it is not possible to rank the demeaned SC and the DID estimators in terms of bias and MSE if treatment assignment is correlated with time-varying common factors. We provide in Appendix \ref{example} a specific example in which the DID can have a smaller bias and MSE relative to the demeaned SC estimator. This might happen when selection into treatment depends on common factors with low variance, and it happens that a simple average of the controls provides a good match for the factor loadings associated with these common factors. In general, however, we should expect a lower bias for the demeaned SC estimator, given that the demeaned SC weights are \emph{partially} chosen to minimize the distance between $\mu_0$ and $\boldsymbol{\mu}' \widehat{\textbf{w}}^{\mbox{\tiny SC$'$}}$, while the DID estimator uses weights that are not data driven. 

Since the biases of these two estimators would generally differ when  $\mathbb{E}[\lambda_t ] \neq 0$ for $t \in \mathcal{T}_1$, we can consider a specification test by contrasting these two estimators. More specifically, if the DID estimator is very different from the demeaned SC estimator, this would suggest that  both estimators are biased (considering the setting in which the pre-treatment fit is imperfect).

A potential problem in  properly testing the equality of these two estimators is that they are generally not asymptotically normal. Still, if we consider a stronger assumption that $\lambda_t$ and $\epsilon_{jt}$ are stationary and weakly dependent for all $t \in \mathcal{T}_0 \cup \mathcal{T}_1$, which implies that  $\mathbb{E}[\lambda_t] = 0$ for $t\in \mathcal{T}_1$, then we can follow the idea from \cite{Chernozhukov} and test this condition using in-time placebos.  More specifically, let $\widetilde{\mathbf{w}}$ be the demeaned SC weights using all periods. We consider
\begin{eqnarray}
\hat u_t =\left( \widetilde{\mathbf{w}}'\mathbf{y}_t + \frac{1}{T_0 + T_1} \sum_{\tau \in \mathcal{T}_0 \cup \mathcal{T}_1 } (y_{0\tau} -  \widetilde{\mathbf{w}}'\mathbf{y}_\tau) \right) - \left( J^{-1} \mathbf{i}'\mathbf{y}_t + \frac{1}{T_0 + T_1} \sum_{\tau \in \mathcal{T}_0 \cup \mathcal{T}_1 } (y_{0\tau} -  J^{-1} \mathbf{i}'\mathbf{y}_\tau) \right).
\end{eqnarray}

Following \cite{Chernozhukov}, we impose the null  $\mathbb{E}[\lambda_t] = 0$ for $t\in \mathcal{T}_1$ and estimate the model using all periods of data to provide better finite sample properties. 

The main idea is that, under the null that $\lambda_t$ and $\epsilon_{jt}$ are stationary and weakly dependent, $\hat u_t$ will approximately be stationary and weakly dependent. Therefore, we can construct a test statistic $S(\widehat{\mathbf{u}}) = \left| \frac{1}{T_1} \sum_{t \in \mathcal{T}_1} \hat u_t \right|$, and derive the distribution of the test statistic by considering the set of all moving block permutations of the time periods. Let $\mathcal{T}_0=\{1,...,T_0\}$ and $\mathcal{T}_1=\{T_0+1,...,T\}$, and let $\Pi$ be the set of permutations $\pi_j$ indexed by $j \in {0,...,T-1}$ such that
\begin{eqnarray}
\pi_j(i) = \begin{cases} i+j & \mbox{ if }  i+j \leq T \\ i+j-T & \mbox{otherwise}.  \end{cases}
\end{eqnarray}

Then the p-value of the specification test is given by
\begin{eqnarray}
\hat p = \frac{1}{T} \sum_{j=0}^{T-1} \mathbf{1} \left\{  S(\widehat{\mathbf{u}}) > S(\widehat{\mathbf{u}}_{\pi_j})  \right\},
\end{eqnarray}

We formalize this idea in the following proposition.

\begin{proposition} \label{Prop_specification}

Assume $\lambda_t$ and $\epsilon_{jt}$ are stationary and weakly dependent, with finite second moments, and that Assumptions \ref{assumption_LFM} and \ref{assumption_sample} hold. Then, for any $a \in (0,1)$, $Pr(\hat p \leq a) \rightarrow a$ when $T_0 \rightarrow \infty$ and $T_1$ is fixed.

\end{proposition}

If we find a low $\hat p$, indicating that the DID and the  demeaned SC estimators are significantly different, then this would be an indication that both estimators are biased (considering the case in which the pre-treatment fit is imperfect). In contrast, a high $\hat p$  would provide some evidence that the condition $\mathbb{E}[\lambda_t] = 0$  for $t\in \mathcal{T}_1$ is valid. If $\lambda_t$ and $\epsilon_{jt}$ are serially uncorrelated, then this test is exact. If there is serial correlation, though, then we may have distortions when $T_0$ is finite, but the test is asymptotically valid when $T_0 \rightarrow \infty$.

Importantly,  this test is completely uninformative about Assumption \ref{assumption_exogeneity}. Moreover, it relies on stationarity as an auxiliary assumption, which is not a necessary assumption for validity of the DID and the demeaned SC estimators in this setting.  We also recommend applied researchers should plot the demeaned SC and the DID estimators to provide a visual inspection of the differences between these two estimators.

\begin{remark}
\normalfont \label{Remark_demeaned}
In general, it is not possible to compare the original and the demeaned SC estimators in terms of bias and variance. For example, if units with similar factor loadings also have similar fixed effects, then matching also on the levels would help provide a better approximation to $\mu_0$.  Moreover, the demeaning process may increase the variance of the estimator for a finite $T_0$.  Finally, demeaning essentially implies that we allow for extrapolation, which some may consider as one of the advantage of the SC estimator (e.g., \cite{Abadie2015}). Therefore, it is not clear whether demeaning is the best option in all applications, and the use of this estimator depends on the willingness of the researcher to allow for extrapolation. 

\end{remark}

\begin{remark}
\normalfont \label{Remark_other}
Our main result that the original and the demeaned SC estimators are generally asymptotically biased if there are unobserved time-varying confounders (Propositions \ref{main_result} and \ref{Prop_demeaned}) still applies if we also relax the non-negative and the adding-up constraints, which essentially leads to the panel data approach suggested by \cite{Hsiao}, and further explored by \cite{Li}. Our conditions for unbiasedness of the SC estimator also apply to the estimators proposed by  \cite{Carvalho2015} and  \cite{Carvalho2016b} when $J$ is fixed. In Appendix  \ref{relaxing_constraints} we show that these papers rely on assumptions that implicitly imply that there is no selection on time-varying unobservables. This clarifies  what selection on unobservables means in this setting, and reconciles our findings with the asymptotic unbiasedness/consistency results in these papers.
 \end{remark}

\subsection{Model with ``diverging'' common factors} \label{Setting2}

While the assumptions considered in Sections \ref{Bias_SC} and \ref{Sub_demeaned} allow for outcomes with divergent pre-treatment averages (which would be the case when we consider, for example, GDP or average wages), we restrict to settings in which such diverging common shocks affect all units in the same way.  In Appendix \ref{Appendix_diverging} we modify Assumption \ref{assumption_LFM} to consider the case in which we may have diverging  common shocks with heterogeneous effects across unit. 

Assuming that there exist weights that reconstruct the factor loadings of the treated unit associated to the diverging common shocks, we show that the asymptotic distribution of the \textit{demeaned} SC estimator does not depend on such diverging common shocks when $T_0 \rightarrow \infty$. The intuition is that, as $T_0 \rightarrow \infty$,  the variance of a linear combination of the idiosyncratic shocks becomes irrelevant relative to the cost of failing to recover the factor loadings associated with the diverging common shocks. {This is consistent with the conclusion from Section \ref{Bias_SC} that the bias of the SC estimator should be less relevant when the common shocks are stronger relative to the idiosyncratic shocks.} However, if we also have non-diverging common shocks, then the demeaned SC weights will generally not asymptotically recover the factor loadings of the treated unit associated with those non-diverging shocks. This implies that the demeaned SC estimator may be asymptotically biased if there is correlation between treatment assignment and these non-diverging shocks, for exactly the same reasons outlined in Section \ref{Bias_SC}.

We also show that the conclusion that the demeaned SC estimator does not depend on the diverging common shocks is not valid for the original SC estimator. While the SC weights considering the original SC method converge in probability to weights that recover the factor loadings of the treated associated to the diverging common shocks, this convergence may not be fast enough to compensate that such common shocks are diverging. We present all details on this setting in Appendix \ref{Appendix_diverging}.

\section{Monte Carlo simulations} \label{Sec_MC}

To illustrate our theoretical findings, we construct a MC simulation based on a real dataset using the monthly employment data for 50 US states and the District of Columbia ($J+1 = 51$) from January 1982 to December 2019 ($T = 456$). We construct these series by aggregating the Current Population Survey (CPS) microdata at the state $\times$ month level.\footnote{We created our CPS extract using IPUMS (\cite{ipums}).}  We estimate a factor model that best approximates this data. In addition to the state and time fixed effects, we estimate four common factors.\footnote{We estimate the linear factor model using the iterated fixed effects method proposed by \cite{Bai}.  The number of factors was selected using the $IC_{p1}$ criterion in \cite{baing}.  We used  the \texttt{interFE} function in the package \texttt{gsynth}  \citep{XU}.}  We find evidence that these four factors and the 51 idiosyncratic shocks are stationary, suggesting that any non-stationary trends in the outcomes come from the time fixed effects, $\delta_t$.\footnote{For each of these series, we consider the Augmented Dickey-Fuller test for a unit root, where the number of lags is chosen using the MAIC criterion of  \cite{Ng2001}. We also test for the presence of deterministic trends using  the test statistics proposed by  \cite{Dickey1981}.} 
  This provides evidence that this dataset is well approximated by the setting we consider in Sections \ref{Bias_SC} and \ref{Sub_demeaned}. We consider state-specific Gaussian ARMA models for the distribution of each $\epsilon_{jt}$ and also for the distribution of the four common factors $\lambda_t$.\footnote{For each estimated factor and for each state time series of residuals, we use the \texttt{auto.arima} function in the R package \texttt{forecast} to perform grid search over the autoregressive and moving-average dimensions; and select the best model according to the BIC criterion  \citep{Hyndman2009}.} 

 For each simulation, we fix the estimated factor loadings $\mu_j$ and consider random draws for $\lambda_t$ and $\epsilon_{jt}$  We set $T_1 = 12$ (one year) and $T_0 \in \{ 120,240,480,1200 \}$, and we vary which state is considered as treated. We consider $5000$ replications for each scenario. The case with $T_0=480$  have approximately the same number of pre-treatment periods as we have in our original dataset, while the cases with smaller $T_0$ reflect more common setting in which we have 10 or 20 years of pre-treatment data. We include the case $T_0 = 1200$ to approximate the asymptotic behavior of the estimators when $T_0 \rightarrow \infty$. 

We consider in Panels A and B of Table  \ref{Table_MC} the case in which $\mathbb{E}[\lambda_{t}]= 0$ for $t \in \mathcal{T}_1$. In Panel A we consider that the treated state is such that $c_{0}$ is the second largest value in the distribution of $c_{j}$, while in Panel B the treated state has the second smallest value of $c_0$. We find that the original SC estimator is biased in both cases, even though it would be possible to have weights that reconstruct  $c_{0}$. The problem is that such weights would be very concentrated on the state with largest (or smallest) $c_j$, so the SC weights would converge to weights that are more diluted, even if this means not recovering $c_0$. This is consistent with Proposition \ref{main_result}. 

Figure \ref{Fig_MC}.A shows the bias of the original SC estimator as a function of the state fixed effects of the treated unit. We find relevant bias of SC estimator when the fixed effects of the treated unit is in the extremes, while the bias is closer to zero when the treated unit is more in the middle of the distribution of the fixed effects. This happens because, as we consider a treated unit with fixed effects more towards the center of the distribution of fixed effects, we can have a weighted average of the control units with more diluted weights that reconstruct the fixed effects of the treated. Therefore, the term in equation \ref{Q0} related to the variance of the linear combination of idiosyncratic shocks becomes less relevant in the minimization problem. Indeed, if we consider a measure of concentration of weights given by $||\mathbf{\widehat w}^{\mbox{\tiny SC}}||_2$, then the correlation between the absolute value of the bias of the SC estimator and this measure ranges from 0.723 to 0.849 (depending on the value of $T_0$).\footnote{For each $T_0$ and a treated state, we calculate the average bias and the average $||\mathbf{\widehat w}^{\mbox{\tiny SC}}||_2$. Then we consider the correlation between these two variables for each $T_0$.  }

\begin{table}[h]
  \centering
\caption{{\bf Monte Carlo Simulations}} \label{Table_MC}
      \begin{lrbox}{\tablebox}
\begin{tabular}{cccccccc}
\hline
\hline

&  \multicolumn{3}{c}{Bias }  & &  \multicolumn{3}{c}{Standard error  }     \\ \cline{2-4} \cline{6-8} 

& SC & Demeaned SC & DID && SC & Demeaned SC & DID \\

     & (1) &  (2) & (3) &  & (4) &  (5) & (6)  \\ 
     
     \hline
 
   \multicolumn{8}{c}{ Panel A: second largest $c_j$, no break }\\
$T_0 = 120$ & 0.243 & 0.009 & 0.020 &  & 0.531 & 0.573 & 0.915 \\
$T_0 = 240$ & 0.200 & -0.002 & 0.006 &  & 0.505 & 0.532 & 0.960 \\
$T_0 = 480$ & 0.196 & 0.006 & 0.007 &  & 0.501 & 0.514 & 1.007 \\
$T_0 = 12000$ & 0.167 & -0.004 & -0.010 &  & 0.503 & 0.500 & 1.005 \\
   \multicolumn{8}{c}{Panel B: second smallest $c_j$, no break} \\
$T_0 = 120$ & -0.565 & 0.020 & 0.028 &  & 1.173 & 0.823 & 1.368 \\
$T_0 = 240$ & -0.601 & -0.012 & 0.002 &  & 1.242 & 0.801 & 1.431 \\
$T_0 = 480$ & -0.596 & -0.011 & 0.004 &  & 1.264 & 0.755 & 1.469 \\
$T_0 = 12000$ & -0.581 & -0.002 & -0.003 &  & 1.242 & 0.714 & 1.442 \\
   \multicolumn{8}{c}{Panel C: second largest $\mu_{1j}$, break in factor 1} \\
$T_0 = 120$ & 1.142 & 1.043 & 2.143 &  & 0.772 & 0.819 & 1.247 \\
$T_0 = 240$ & 0.934 & 0.797 & 2.129 &  & 0.728 & 0.764 & 1.308 \\
$T_0 = 480$ & 0.808 & 0.675 & 2.110 &  & 0.713 & 0.728 & 1.337 \\
$T_0 = 12000$ & 0.718 & 0.598 & 2.112 &  & 0.688 & 0.691 & 1.310 \\
   \multicolumn{8}{c}{Panel D: second smallest $\mu_{1j}$, break in factor 1} \\
$T_0 = 120$ & -1.821 & -1.851 & -3.176 &  & 1.200 & 1.167 & 1.777 \\
$T_0 = 240$ & -1.573 & -1.532 & -3.157 &  & 1.130 & 1.113 & 1.794 \\
$T_0 = 480$ & -1.415 & -1.324 & -3.110 &  & 1.082 & 1.058 & 1.816 \\
$T_0 = 12000$ & -1.366 & -1.230 & -3.162 &  & 1.028 & 0.996 & 1.793 \\

\hline
\hline

\end{tabular}
   \end{lrbox}
\usebox{\tablebox}\\
\settowidth{\tableboxwidth}{\usebox{\tablebox}} \parbox{\tableboxwidth}{\footnotesize{Notes: this table presents the MC simulations discussed in Section \ref{Sec_MC}. Panels A and B consider the case in which all common factors have mean zero in the post-treatment periods. In Panel A, the treated unit is the state with second largest fixed effect in the distribution of $c_j$, while in Panel B the treated unit is the state with the second smallest fixed effect. Panels C and D consider the case in which the first common factor has expected value equal to two times its standard deviation in the post-treatment periods.   In Panel C, the treated unit is the state with second largest factor loadings associated to the first common factor  in the distribution of $\mu_{1j}$, while in Panel D the treated unit is the state with the second smallest factor loading. In all simulations, the true treatment effect is equal to zero.
}
}
\end{table}

 \begin{figure}[h] 
\centering 
\caption{{\bf Monte Carlo Simulations}} \label{Fig_MC}
\bigskip

\begin{tabular}{cc}

 \ref{basque}.A: Bias of original SC &   \ref{basque}.B:  Bias of demeaned SC \\
\includegraphics[scale=0.5]{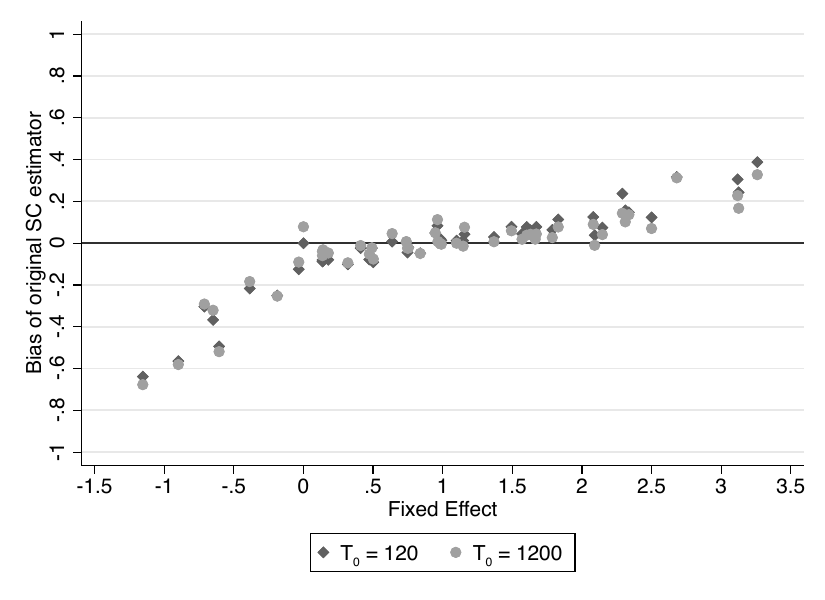} &   \includegraphics[scale=0.5]{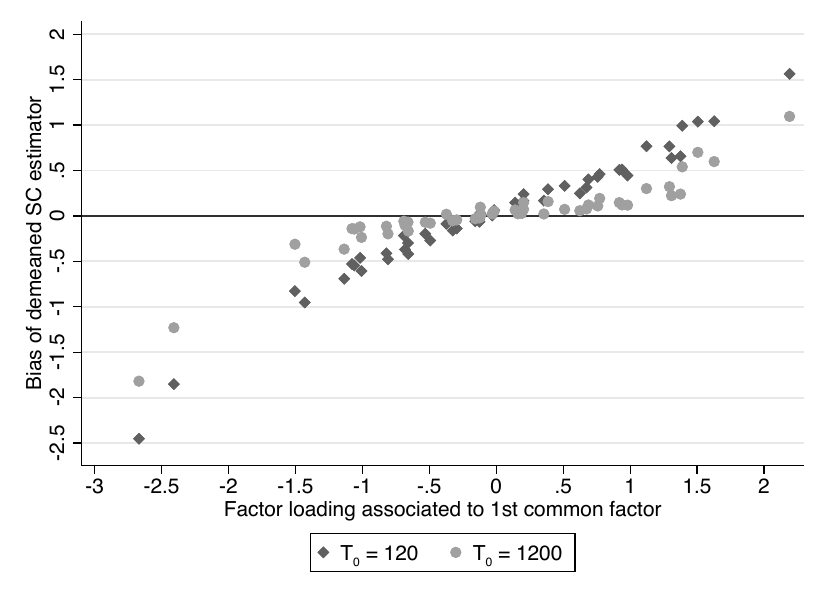} \\

\end{tabular}
\fnote{Notes: Figure A presents the bias of the original SC estimator as a function of the fixed effect of the treated state. We consider in this case a setting with no structural break for the common factors. Figure B presents the bias of the demeaned SC estimator as a function of the factor loadings associated to the first common factor of the treated state. We consider the case in which  the first common factor has expected value equal to two times its standard deviation in the post-treatment periods. We present the settings with $T_0=120$ and $T_0=1200$.  In all simulations, the true treatment effect is equal to zero. }
\end{figure}

In contrast to the original SC estimator, both the demeaned SC and the DID estimators are unbiased in this setting, as presented in Panel A of Table  \ref{Table_MC}. This is consistent with Proposition \ref{Prop_demeaned}. Moreover, as expected from Proposition \ref{Prop_demeaned_efficient}, the demeaned SC estimator is more efficient than the DID estimator, with  roughly 40\%-50\% smaller standard errors. 

In Panels C and D of Table \ref{Table_MC}, we consider a setting in which the first common factor has expected value equal to two times its standard deviation in the post-treatment periods.  In Panel C, we consider the case in which the treated state is such that $\mu_{10}$ is the second largest value in the distribution of $\mu_{1j}$, while in Panel D we consider the case in which it is the second smallest. Therefore, again we are in a setting in which it would be possible to construct a SC state that is affected by $\lambda_{1t}$ in the same way as the treated unit. Still, the results from columns 1 and 2 show that the original and the demeaned SC estimators are biased even when $T_0$ is large. This happens because the SC weights fail to reconstruct  $\mu_{10}$, despite the fact that there exist weights that would do so. This is consistent with the results from Propositions \ref{main_result} and  \ref{Prop_demeaned}. Note also that the biases of the original and demeaned SC estimators are larger when $T_0$ is smaller. See \cite{FP_old} for a more thorough discussion on that. Figure \ref{Fig_MC}.B shows the bias of the demeaned SC estimator as a function of the factor loadings of the treated unit. Again, the bias is closer to zero when the treated unit is in the middle of the distribution of $\mu_{1j}$.  If we consider a measure of concentration of weights given by $||\mathbf{\widehat w}^{\mbox{\tiny SC}'}||_2$, then the correlation between the absolute value of the bias of the demeaned SC estimator and this measure ranges from 0.723 to 0.810 (depending on the value of $T_0$).

While the biases of the original and the demeaned SC estimators do not converge to zero when there is selection on unobservables, these biases are substantially smaller than the bias of the DID estimator in this setting (column 3). Interestingly, the bias of the demeaned SC estimator attenuates the bias from the DID estimator, but does not eliminate it. The idea is that, when we move from the DID to the demeaned SC estimator, the weights move in the direction of weights that reconstruct the factor loadings. However, the SC weights do not generally go all the way through to recover $\mu_0$ because of the idiosyncratic shocks. Moreover, the DID estimator presents a substantially larger standard error (columns 4 to 6). These results are consistent with the conclusions from Section \ref{Sub_demeaned}, in that the demeaned SC estimator improves relative to DID in terms of bias and variance.\footnote{In these simulations, the bias of the original SC estimator is only slightly larger than the bias of the demeaned SC estimator, suggesting that, in this setting, the SC state  approximately reconstructs  the  fixed effect of the treated state. Note, however, that such comparison cannot be extrapolated to other settings, as discussed in Remark \ref{Remark_demeaned}.}

Finally, we consider  the size and power of the specification test proposed in Section \ref{Sub_demeaned}. We present in columns 1 to 4 of Appendix Table \ref{Table_spec} rejection rates when there is no structural break (so both the demeaned SC and the DID estimators are asymptotically unbiased), depending on the treated state and on $T_0$. The test presents relevant over-rejection when $T_0$ is small, but such distortions become less relevant when $T_0$ increases. Such distortions with finite $T_0$ arise because the common factors exhibit serial dependence. If we did not have dependence, then the test would be exact. Overall, the fact that the test has some over-rejection when $T_0$ is small is less worrisome than if we had under-rejection,  because this would lead researchers to be more cautious  about the use of the demeaned SC estimator.  In columns 5 to 8, we consider the case in which there is a structural break. In this case, the test would have  power to detect that the demeaned SC and the DID estimators are different, especially when the treated unit is on the extremes of the distribution of $\mu_{1j}$. When the treated unit is in the middle of this distribution, then the bias of both the demeaned SC and of the DID estimators become less relevant, so the probability of rejecting specification test becomes lower.

\section{Empirical Illustration} \label{Sec_EI}

As an empirical illustration, we revisit the application presented by  \cite{Abadie2003}.  We present in Figure \ref{basque}.A the per capita GDP time series for the Basque Country and for other Spanish regions, while in Figure \ref{basque}.B we replicate their Figure 1, which displays the per capita GDP of the Basque Country contrasted with the per capita GDP of a SC unit constructed to provide a counterfactual for the Basque Country without terrorism. We construct three different SC units, with the original SC estimator using all pre-treatment outcome lags as predictors, with the demeaned SC estimator using all pre-treatment outcome lags as predictors, and with the specification considered by \cite{Abadie2003}. All specifications point out to large negative treatment effects, although the estimated effects are slightly smaller for the original specifications considered by  \cite{Abadie2003}.

Figure \ref{basque}.B displays a remarkably good pre-treatment fit, regardless of the specification. However, the per capita GDP series are clearly non-stationary, with all regions displaying  similar trends before the intervention. Considering the results presented in Section \ref{Sec_main_results}, such non-stationarity may come either from time fixed effects $\delta_t$, or from non-stationary common shocks that may have heterogeneous effects across regions. If the non-stationarity comes from a common factor $\delta_t$ that affects every unit in the same way, then the series $\tilde y_{jt} = y_{jt} - \frac{1}{J}\sum_{j' \neq 0}y_{j't}$ would not display non-stationary trends. As shown in Figure \ref{basque}.C, this appears to be the case in this application.\footnote{If there were other sources of non-stationarity, then the series would remain non-stationary even after such transformation. In such cases, other strategies to de-trend the series could be used, such as, for example, considering parametric trends.}   
  
In light of the results from Section \ref{Sec_main_results}, the distortions in the SC weights depend on the relative magnitudes of the variance of the \textit{non-diverging} common factors relative to the variance of the idiosyncratic shocks. Therefore,  Figure \ref{basque}.C provides a better visual assessment of whether the pre-treatment fit is good relative to Figure \ref{basque}.B. While the pre-treatment fit is still reasonably good after we discard the non-stationary part of the series, it is not as good as when we consider the series in levels. 

We can consider whether we would be able to justify the use of the SC method in this setting without relying on theoretical results based on perfect pre-treatment fit approximations. First, note that the estimated weights in this application are very concentrated among a few control regions, so we cannot rely on the theoretical results from \cite{Ferman} to argue that the SC estimator is asymptotically unbiased in this setting (see Appendix Table \ref{Appendix_table_weights}). Following the discussion in Section  \ref{Sub_demeaned}, we also contrast  in Figure \ref{basque}.D the DID and the demeaned SC estimators. If they were very similar, then we would have some support to rely on these estimators even if they failed to reconstruct the factor loadings of the treated region. However, we find the the estimated effect using the demeaned SC estimator is systematically larger. The p-value of the specification test proposed in Proposition \ref{Prop_specification} is $0.023$.  This suggests that we may have selection on time-varying unobservables, implying that both the demeaned and the DID estimators are asymptotically biased, although the bias of the demeaned SC estimator should be smaller.

 \begin{figure}[h] 
\centering 
\caption{{\bf \cite{Abadie2003} application}} \label{basque}

\begin{tabular}{cc}

 \ref{basque}.A:  Raw data &   \ref{basque}.B:  GDP in level \\
\includegraphics[scale=0.25]{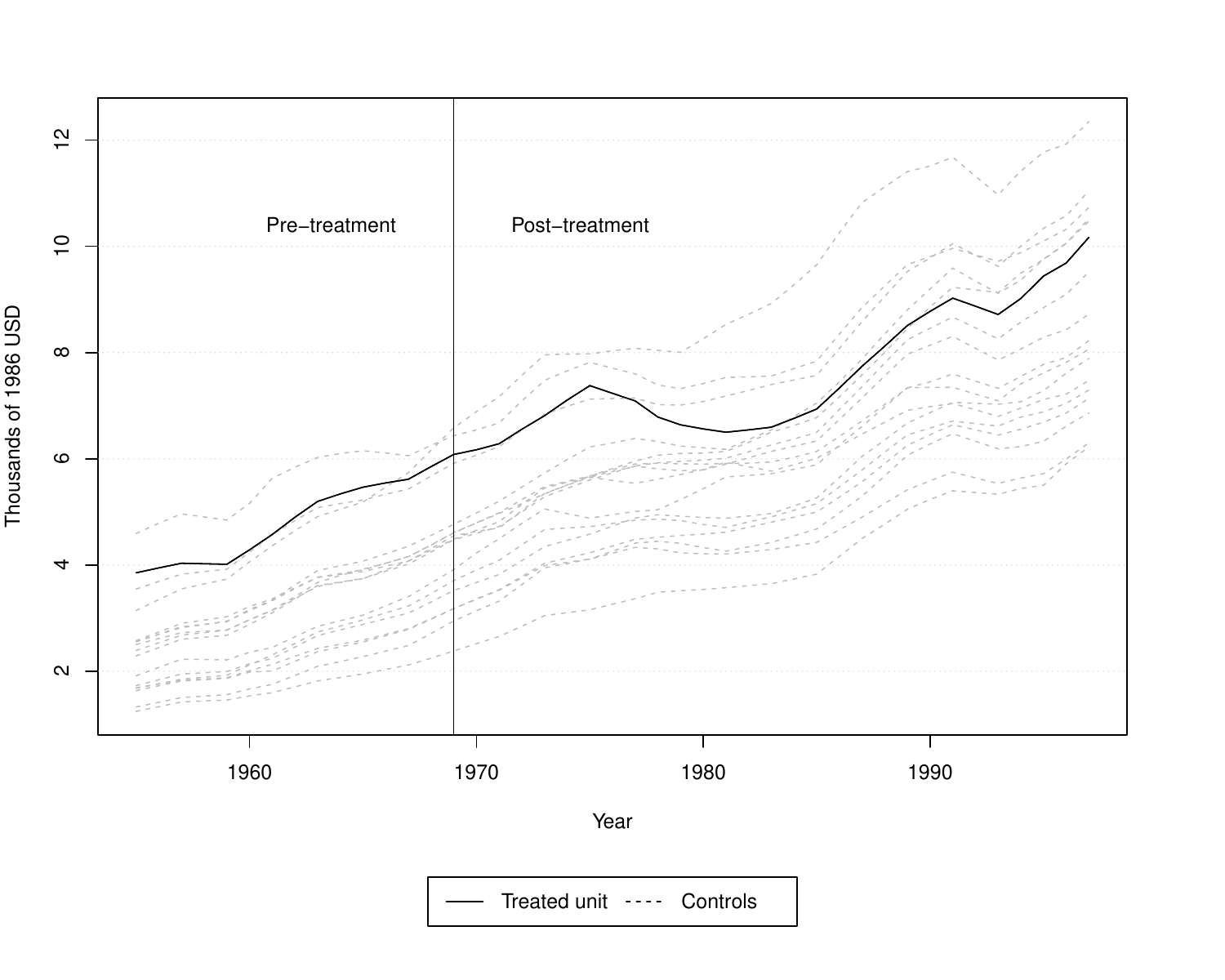} &   \includegraphics[scale=0.25]{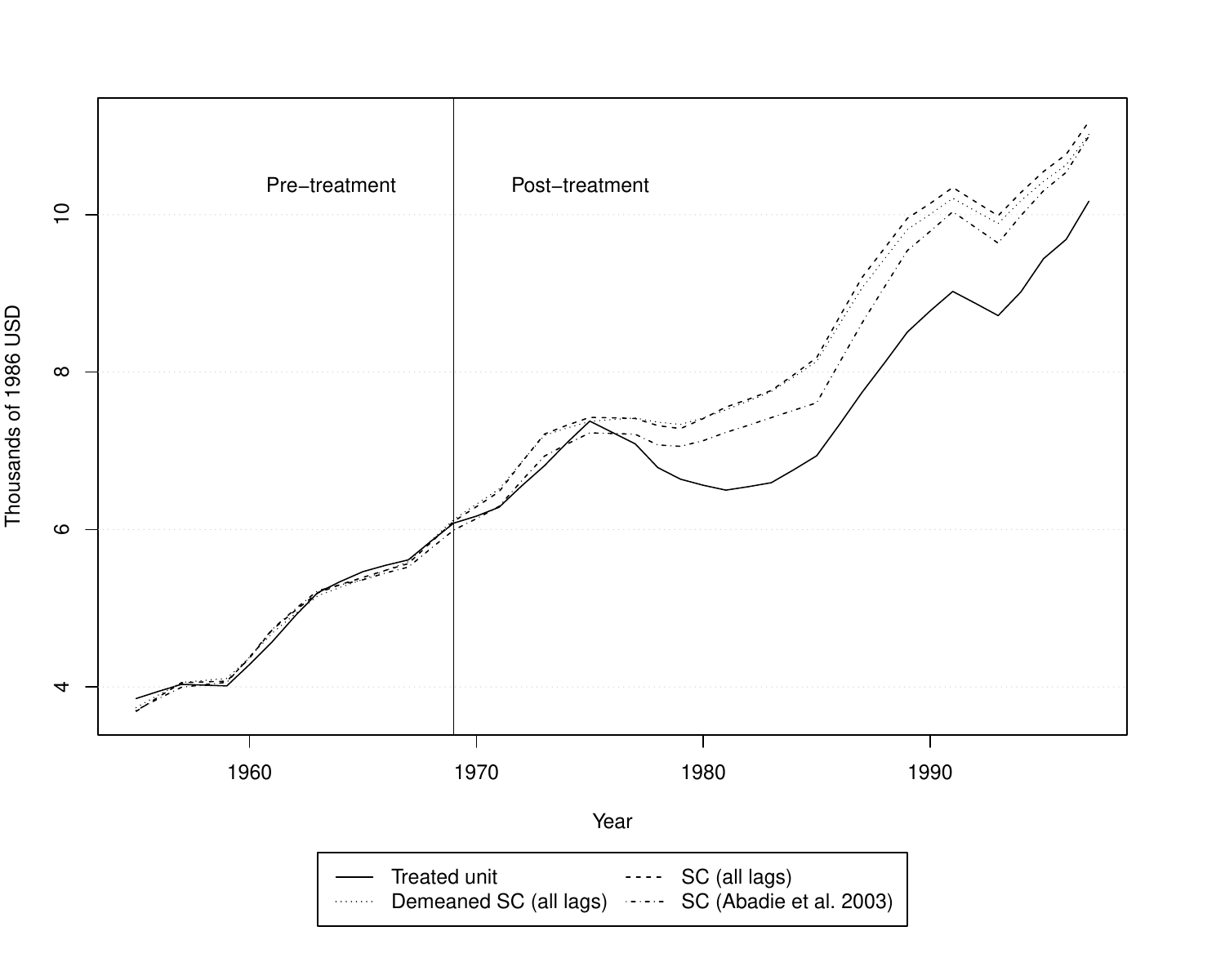} \\

  \ref{basque}.C: GDP de-trended &   \ref{basque}.D: Demeaned SC vs DID estimators  \\
\includegraphics[scale=0.25]{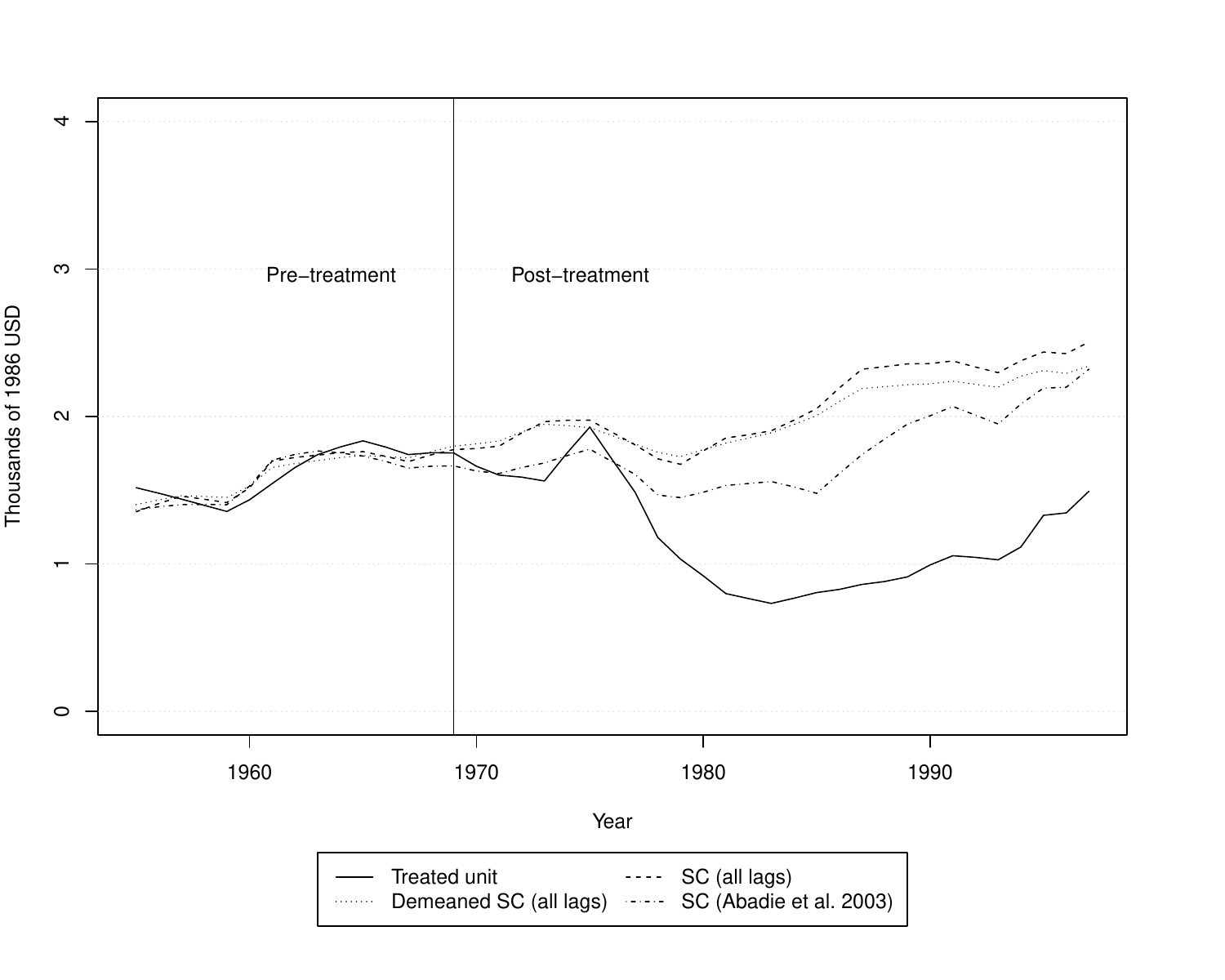}  & \includegraphics[scale=0.25]{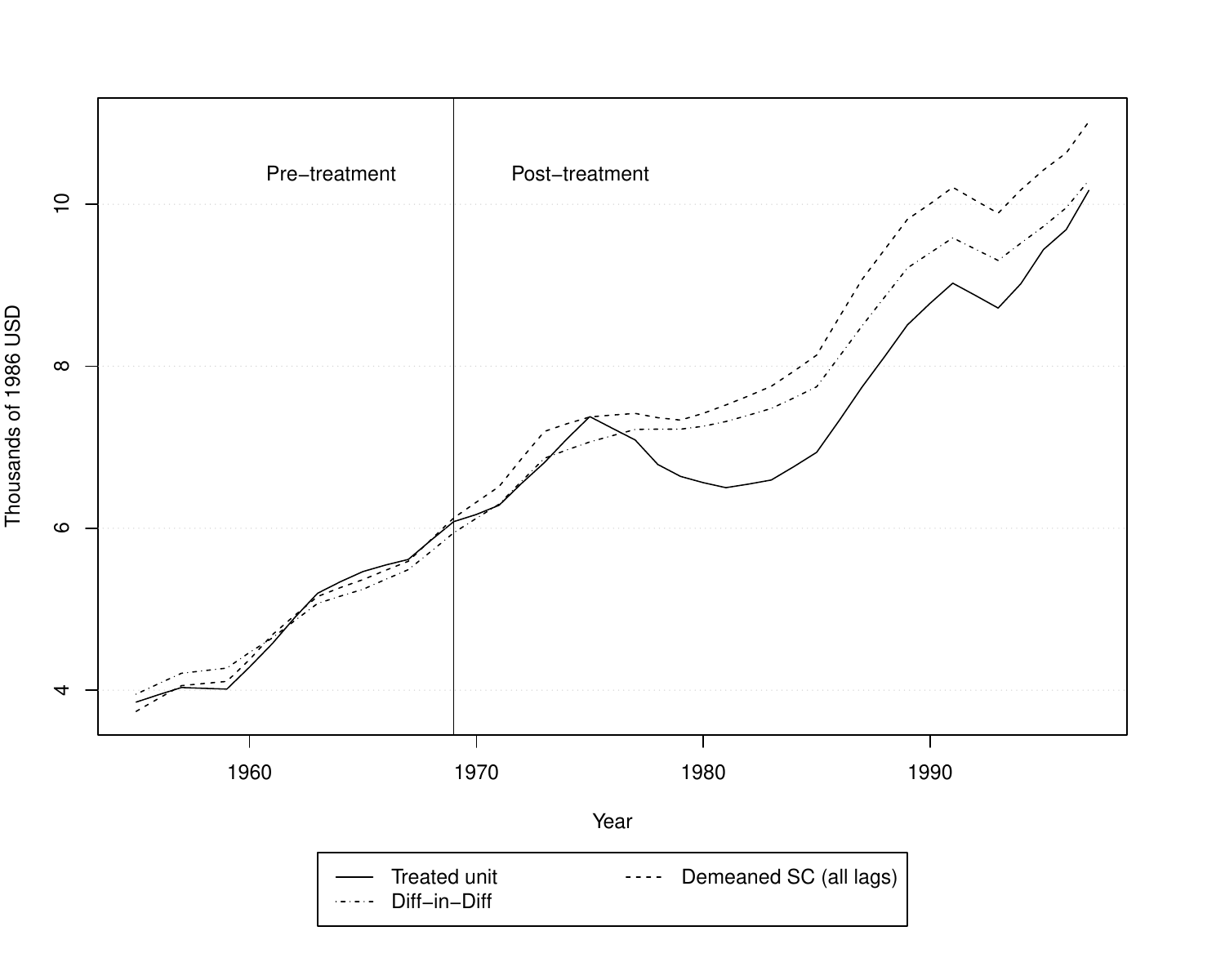} 

\end{tabular}
\fnote{Notes: Figure A presents time series for  the treated and for the control units used in the empirical application from \cite{Abadie2003}. In Figure B we present the time series for the treated and for the SC units. We consider the SC unit estimated with the original SC estimator using all pre-treatment periods lags, with the demeaned SC estimator using all pre-treatment periods lags, and with the specification considered by \cite{Abadie2003}. In Figure C we present the same information as in Figure B after subtracting the control groups' averages for each time period.  In Figure D we present the counterfactuals using the demeaned SC and the DID estimators. }
\end{figure}

Overall, since in this particular application the pre-treatment fit is reasonably good even once we subtract the non-stationary trends, and the treatment effects are large relative to the pre-treatment gaps, we should expect that any potential bias from the demeaned SC estimator does not explain a large proportion of the estimated effects.  Moreover, given the discussions from Sections \ref{Sub_demeaned} and \ref{Sec_MC}, we should expect  the demeaned SC estimator to \emph{partially} control for any bias that the DID estimator experience. Since the estimated effects with the demeaned SC estimator are stronger than the DID estimates, given this rationale, we should expect, if anything, that the demeaned SC estimator would provide a lower bound on the (absolute values of the) treatment effects. Therefore, a careful analysis of the potential problems of the SC method would not change the main conclusions from this empirical application. Still, in other settings in which the pre-treatment fit is worse, and in which moving from the DID to the demeaned SC estimator leads to weaker results, then it would not be possible to rely on the arguments used above, and the problems we highlight in this paper may undermine conclusions from the SC method.

\section{Recommendations} \label{recommendations}

Taken together, our results clarify the conditions in which the SC and related estimators can be reliably used, and provide guidance on how applied researchers could justify the use of these methods. First, a condition like the one we present in Assumption \ref{assumption_exogeneity} is always necessary to justify the SC estimator. It states that treatment assignment is not related to shocks that are specific to the treated unit. It does allow, however, for unobserved confounders that may also affect other control units. Indeed, the main reason why a researcher should use these kind of methods is if he/she believes that there may be confounding factors that also affect the control units. In this case, information from the control units could be used to control for such confounders. Therefore, any applied paper relying on the SC method should discuss the possible unobserved confounders in the specific application, and argue that such confounders are not specific to the treated unit. 

Importantly, even if it is plausible that idiosyncratic shocks are not correlated with the treatment assignment, whether the SC method is able to reliably control for the common shocks depends crucially on details of the empirical application. There are two settings that provide validity for the SC estimator even when there are time-varying unobserved confounders. First, \cite{Abadie2010} show that the SC estimator is reliable if the pre-treatment fit is good for a large number of pre-treatment periods. This condition can be checked by contrasting the outcomes of the treated and of the SC units in the pre-treatment periods. Based on our results, we recommend that applied researchers should also consider the pre-treatment fit after discarding diverging trends, in order to provide a better understanding of the relative magnitude between the variances of the non-diverging common factors and of the idiosyncratic shocks. Also, it is important that the number of control units in this case cannot be large in comparison to $T_0$, otherwise a good pre-treatment fit might be a consequence of over-fitting. In this case, the bias of the SC estimator we uncover in our paper may remain relevant even if we have a good pre-treatment fit.

Second, when both $J$ and $T_0$ are large,  \cite{Ferman} show that the SC estimator may be asymptotically unbiased even when the pre-treatment fit is imperfect. This would be the case if the confounders affect a large number of control units, and in this case the SC weights would get diluted among an increasing number of control units when $J \rightarrow \infty$.   Therefore, we recommend that applied researchers also  report the $L_2$ norm of the SC weights, $||\widehat{\textbf{w}}^{\mbox{\tiny SC}}||_2$. If this is close to zero, then we would have evidence that we are closer to the setting considered by \cite{Ferman}. In contrast, if the weights are concentrated, then we would have evidence that the bias we uncover in our paper is potentially relevant. As we show in our MC simulations in Section \ref{Sec_MC}, the cases in which we find largest biases are exactly the ones in which the SC weights are more concentrated.  

The results we derive in Section \ref{Sec_main_results} are informative about the properties of the SC estimator when the conditions outlined by \cite{Abadie2010} and \cite{Ferman} do not hold. This would be the case when (i) the pre-treatment fit is imperfect and $J$ is not large, (ii) the pre-treatment fit is imperfect with large  $J$ and $T_0$, but the SC weights are not diluted among a large number of control units, or (iii)  the pre-treatment fit is good, but $J$ is much larger than $T_0$, so such pre-treatment fit is possibly good due to over-fitting. 

In these cases, we show that the SC estimator can still provide important gains relative to the DID estimator, but the applied researcher should be more careful in justifying the use of the method. If one considers the demeaned SC estimator, then the assumptions for unbiasedness would be the same as those for the DID estimator. That is, the researcher should argue that the relevant unobserved confounders are not time-varying. The advantage of relying on the demeaned SC estimator relative to DID in this case is that it would be more efficient if common shocks are stable before and after the treatment, and that it should have lower bias in case there is correlation between treatment assignment and time-varying unobservables. We also show that contrasting the DID and the demeaned SC estimators is informative about whether these conditions for unbiasedness are valid, and propose a specification test based on that. If we find evidence that these two estimators are similar, then we should be more confident that the conditions for asymptotic unbiasedness of the demeaned SC estimator holds even when we are not in the settings considered by  \cite{Abadie2010} or \cite{Ferman}. 

Importantly, if the conditions considered by \cite{Abadie2010} or \cite{Ferman} hold in a specific application, then the demeaned SC estimator would be asymptotically unbiased, while the DID estimator may be biased. In this case,  an information from the specification test indicating that the demeaned SC and the DID estimators are different would not imply that the demeaned SC estimator is invalid. Therefore, it is crucial to understand the conditions under which each of these estimators are valid to interpret the conclusions from this specification test. 

Finally, if one considers the original SC estimator, then one should have to inspect the pre-treatment fit. If the SC unit recovers the levels of the treated unit (even if the pre-treatment fit is imperfect), then again the estimator would be reliable if there is no relevant time-varying unobserved confounders. If the SC unit does not recover the levels, then the original SC estimator should not be used. In such settings, the researcher should either use the demeaned SC estimator, or discard such application in case he/she does not want to rely on extrapolation.

\section{Conclusion  } \label{Conclusion}

We consider the properties of the SC and related estimators, in a linear factor model setting, when the pre-treatment fit is imperfect. We show that, in this framework, the SC estimator is generally biased if treatment assignment is correlated with the unobserved heterogeneity, and that such bias does not converge to zero even when the number of pre-treatment periods is large. Still, we  also show that a modified version of the SC method can substantially improve relative to DID, even if the pre-treatment fit is not close to perfect and if $T_0$ is not large.  Overall, we show that the SC method can provide substantial improvement relative to DID, even in settings where the method was not originally designed to work. However, researchers should be more careful in the evaluation of the identification assumptions in those cases. {Importantly, our results clarify the conditions in which the SC and related estimators are reliable, and provide practical guidance on how applied researchers should justify the use of such estimators in empirical applications. }

\pagebreak


\appendix

\setcounter{table}{0}
\renewcommand\thetable{A.\arabic{table}}

\setcounter{figure}{0}
\renewcommand\thefigure{A.\arabic{figure}}

\pagebreak

\onehalfspacing
\normalsize

\section{Supplemental Appendix: Revisiting the Synthetic Control Estimator (For Online Publication)}

\subsection{Proof of the Main Results} 

\subsubsection{Proposition \ref{main_result} } \label{Prop1}

\begin{proof}

The SC weights $\mathbf{\widehat w} \in \mathbb{R}^J$ are given by\footnote{If the number of control units is greater than the number of pre-treatment periods, then the solution to this minimization problem might not be unique. However, since we consider the asymptotics with $T_0 \rightarrow \infty$, then we guarantee that, for large enough $T_0$, the solution will be unique.}

\begin{eqnarray}
\widehat{\textbf{w}} =  \underset{{\textbf{w} \in \Delta^{J-1}}}{\mbox{argmin}}  \frac{1}{T_0} \sum_{t \in \mathcal{T}_0} \left(y_{0t}-\mathbf{y}_{t}'\mathbf{w}   \right)^2. 
\end{eqnarray}

Under Assumptions \ref{assumption_LFM},  \ref{assumption_sample} and  \ref{assumptions_lambda}, the objective function $\widehat Q_{T_0}(\mathbf{w}) \equiv \frac{1}{T_0} \sum_{t \in \mathcal{T}_0} \left(y_{0t}-\mathbf{y}_{t}'\mathbf{w}   \right)^2  $  converges pointwise in probability  to
\begin{eqnarray} \label{obfunction}
Q_0(\mathbf{w}) \equiv \sigma_\epsilon^2(1+\textbf{w}^\prime\textbf{w})+ \left[\left(\mu_{0}-\boldsymbol{\mu}'\textbf{w}\right)^\prime\Omega_0\left(\mu_{0}-\boldsymbol{\mu}' \textbf{w}\right) + (c_0 - \mathbf{c}'\mathbf{w})^2 \right]
\end{eqnarray}
which is a continuous and strictly convex function. Therefore, $Q_0(\mathbf{w})$ is uniquely minimized over $\Delta^{J-1}$, and we define its minimum as $\mathbf{\bar w} \in \Delta^{J-1}$.

We show that this convergence in probability is uniform over $\mathbf{w} \in \Delta^{J-1}$. Define $\tilde y_{0t} = y_{0t} - \delta_t$  and $ \mathbf{\tilde y}_{t} = \mathbf{y}_{t} - \delta_t \mathbf{i}$, where $\mathbf{i}$ is a $J \times 1$ vector of ones. For any $\mathbf{w}',\mathbf{w} \in \Delta^{J-1}$, using the mean value theorem, we can find a $\mathbf{\widetilde w} \in \Delta^{J-1}$ such that
\begin{eqnarray} \nonumber
\left|  \widehat Q_{T_0}(\mathbf{w}') - \widehat Q_{T_0}(\mathbf{w})  \right| &=& \left | 2 \left( \frac{1}{T_0} \sum_{t \in \mathcal{T}_0} \mathbf{\tilde y}_{t} \tilde y_{0t}  -  \frac{1}{T_0} \sum_{t \in \mathcal{T}_0} \mathbf{\tilde y}_{t}\mathbf{\tilde y}_{t}' \mathbf{\widetilde w} \right) \cdot \left( \mathbf{w}' - \mathbf{w} \right) \right| \\
&\leq & \left[ \left( 2 \left| \left|  \frac{1}{T_0} \sum_{t \in \mathcal{T}_0} \mathbf{\tilde y}_{t} \tilde y_{0t}   \right| \right|  + \left| \left|   \frac{1}{T_0} \sum_{t \in \mathcal{T}_0} \mathbf{\tilde y}_{t}\mathbf{\tilde y}_{t}'  \right| \right| \times  \left| \left|\mathbf{\widetilde w}   \right| \right| \right)   \left| \left| \mathbf{w}' - \mathbf{w} \right| \right|   \right].
\end{eqnarray}

Define  $B_{T_0} = 2 \left| \left|  \frac{1}{T_0} \sum_{t \in \mathcal{T}_0} \mathbf{\tilde y}_{t}\tilde y_{0t}   \right| \right|  + \left| \left|   \frac{1}{T_0} \sum_{t \in \mathcal{T}_0} \mathbf{\tilde y}_{t}\mathbf{\tilde y}_{t}'  \right| \right| \times  C$. Since $\Delta^{J-1}$ is compact, $\left| \left| \mathbf{\widetilde w} \right| \right|$ is bounded, so we can find a constant $C$  such that $\left|  \widehat Q_{T_0}(\mathbf{w}') - \widehat Q_{T_0}(\mathbf{w})  \right| \leq B_{T_0} \left( \left| \left| \mathbf{w}' - \mathbf{w} \right| \right| \right)^{\frac{1}{2}}$. Since  $\tilde y_{0t}\mathbf{\tilde y}_{t} $  and $ \mathbf{\tilde y}_{t}\mathbf{\tilde y}_{t}' $ are linear combinations of cross products of $\lambda_t$ and $\epsilon_{it}$, from Assumptions  \ref{assumption_LFM},   \ref{assumption_sample}, and  \ref{assumptions_lambda}, we have that $B_{T_0}$ converges in probability to a positive constant, so $B_{T_0}=O_p(1)$. Note also that $Q_0(\mathbf{w})$ is uniformly continuous on $\Delta^{J-1}$. Therefore, from Corollary 2.2 of \cite{Newey}, we have that $\widehat Q_{T_0}$ converges uniformly in probability to $Q_0$. Since $Q_{0}$ is uniquely minimized at $\mathbf{\bar w}$, $\Delta^{J-1}$ is a compact space, $Q_{0}$ is continuous and $\widehat{Q}_{T_0}$ converges uniformly to $Q_{{0}}$, from  Theorem $2.1$ of  \cite{NeweyMac1994}, $\widehat{\mathbf{w}}$ exists with probability approaching one, and $\widehat{\mathbf{w}} \buildrel p \over \rightarrow \mathbf{\bar w}$.

Now we  show that $\mathbf{\bar w}$ does not generally reconstruct the factor loadings. Note that $Q_0$ has two parts. The first one reflects that different choices of weights will generate different weighted averages of the idiosyncratic shocks $\epsilon_{it}$. In this simpler case, this part would be minimized when we set all weights equal to $\frac{1}{J}$. Let the $J \times 1$ vector $\mathbf{j_J}=\left(\frac{1}{J},...,\frac{1}{J}\right)' \in \Delta^{J-1}$.  The second part reflects the presence of  common factors $\lambda_t$ and of the unit fixed effects that would remain after we choose the weights to construct the SC unit. This part is minimized if we choose a $\mathbf{w}^\ast \in \widetilde \Phi $.  Suppose that we  start at $\mathbf{w}^\ast \in \Phi$ and move in the direction of  $\mathbf{j_J} $, with  $\mathbf{w}(\Delta) = {\mathbf{w}^\ast}+ \Delta(\mathbf{j_J} - {\mathbf{w}^\ast})$. Note that, for all $\Delta \in [0,1]$, these weights will continue to satisfy the constraints of the minimization problem. If we consider the derivative of function \ref{obfunction} with respect to $\Delta$ at $\Delta=0$, we have that
\begin{eqnarray} \nonumber 
\Gamma'( \mathbf{w}^\ast ) = 2 \sigma_\epsilon^2 \left( \frac{1}{J} -  {\mathbf{w}^\ast}'\mathbf{w}^\ast    \right)  < 0 \mbox{ unless } {\mathbf{w}^\ast}=\mathbf{j_J} \mbox{ or } \sigma_\epsilon^2=0,
\end{eqnarray}
where we used the fact that $\mathbf{j_J}  ' \mathbf{w}^\ast = \frac{1}{J}$, because weights are restricted to sum one.

Therefore, $\textbf{w}^\ast$ will not, in general, minimize $Q_0$. This implies that,  when $T_0 \rightarrow \infty$,  the SC weights will converge in probability to weights  $\mathbf{\bar w}$ that does not reconstruct the factor loadings of the treated unit, unless it turns out that   $\textbf{w}^\ast$ also minimizes the variance of this linear combination of the idiosyncratic errors or if $\sigma^2_\epsilon=0$.

Now considering the SC estimator,
\begin{eqnarray} 
\hat \alpha^{\mbox{\tiny SC}}_{0t} &=& y_{0t} - \mathbf{y}_{t} \mathbf{\widehat w}^{\mbox{\tiny SC}} \buildrel p \over \rightarrow   \alpha_{0t} + \left( \epsilon_{0t} - \boldsymbol{\epsilon}_{t}' \mathbf{\bar w} \right) + \lambda_t \left(\mu_0 -  \boldsymbol{\mu}' \mathbf{\bar w}  \right) + \left(c_0 -  \mathbf{c}' \mathbf{\bar w}  \right).
\end{eqnarray}
\end{proof}

\subsubsection{Proposition \ref{Prop_demeaned} } \label{Proof_demeaned}

\begin{proof}

The demeaned SC estimator is given by $\widehat{\textbf{w}}^{\mbox{\tiny SC$'$}} =  \underset{{ \textbf{w} \in \Delta^{J-1}}}{\mbox{argmin}} ~ \widehat Q_{T_0}'(\mathbf{w})$, where
\begin{eqnarray} \nonumber
 \widehat Q_{T_0}'(\mathbf{w}) &=& \frac{1}{T_0} \sum_{t \in \mathcal{T}_0} \left( y_{0t} - \mathbf{y}_{t}'\mathbf{w}  - \left( \frac{1}{T_0} \sum_{t' \in \mathcal{T}_0} y_{0t'} - \frac{1}{T_0} \sum_{t' \in \mathcal{T}_0}  \mathbf{y}_{t'}'\mathbf{w}  \right)  \right)^2 \\ &
  =& \widehat Q_{T_0}(\mathbf{w}) - \left(  \frac{1}{T_0} \sum_{t \in \mathcal{T}_0} y_{0t} - \frac{1}{T_0} \sum_{t \in \mathcal{T}_0}  \mathbf{y}_{t}'\mathbf{w}    \right)^2.
 \end{eqnarray}

 $ \widehat Q_{T_0}'(\mathbf{w}) $ converges pointwise in probability to
\begin{eqnarray}
  Q_{0}'(\mathbf{w})  \equiv \sigma_\epsilon^2(1+\textbf{w}^\prime\textbf{w})+\left(\mu_{0}-\boldsymbol{\mu}'\textbf{w}\right)^\prime  \Omega   \left(\mu_{0}-\boldsymbol{\mu}'\textbf{w}\right)
\end{eqnarray}
where  $ \Omega_0  - \omega_0'\omega_0$ is positive semi-definite, so $  Q_{0}'(\mathbf{w}) $ is a continuous and convex function.

The proof that $\widehat{\textbf{w}}^{\mbox{\tiny SC$'$}}  \buildrel p \over \rightarrow \mathbf{\bar w}^{\mbox{\tiny SC$'$}}$ where $ \mathbf{\bar w}^{\mbox{\tiny SC$'$}}$ will generally not reconstruct the factor loadings of the treated unit follows exactly the same steps as the proof of Proposition \ref{main_result}. Therefore
\begin{eqnarray} 
\hat \alpha^{\mbox{\tiny SC$'$}}_{0t} &=& y_{0t} - \mathbf{y}_{t} \mathbf{\widehat w}^{\mbox{\tiny SC$'$}}  - \left[ \frac{1}{T_0} \sum_{t' \in \mathcal{T}_0}  y_{0t} - \frac{1}{T_0} \sum_{t' \in \mathcal{T}_0}  \mathbf{y}_{t}'  \mathbf{\widehat w}^{\mbox{\tiny SC$'$}}  \right]  \\
& \buildrel p \over \rightarrow &  \alpha_{0t} + \left( \epsilon_{0t} - \boldsymbol{\epsilon}_{t}' \mathbf{\bar w}^{\mbox{\tiny SC$'$}} \right) + \lambda_t \left(\mu_0 -  \boldsymbol{\mu}' \mathbf{\bar w}^{\mbox{\tiny SC$'$}}  \right).
\end{eqnarray}
\end{proof}

\subsubsection{Proposition \ref{Prop_demeaned_efficient} } \label{Proof_demeaned_efficient}

\begin{proof}

For any estimator $\hat \alpha_{0t}( \mathbf{\widetilde w}) = y_{0t} - \mathbf{y}_{t} \mathbf{\widetilde w}  - \left[ \frac{1}{T_0} \sum_{t' \in \mathcal{T}_0}  y_{0t} - \frac{1}{T_0} \sum_{t' \in \mathcal{T}_0 }  \mathbf{y}_{t}'  \mathbf{\widetilde w} \right] $ such that $\mathbf{\widetilde w} \buildrel p \over \rightarrow \mathbf{ w}$, we have that, under Assumptions \ref{assumption_LFM} to  \ref{A5},
\begin{eqnarray}
a.var(\hat \alpha_{0t}( \mathbf{\widetilde w})  ) = \sigma_\epsilon^2(1+\textbf{w}^\prime\textbf{w})+\left(\mu_{0}-\boldsymbol{\mu}\textbf{w}\right)^\prime  \Omega \left(\mu_{0}-\boldsymbol{\mu}\textbf{w}\right) =   Q_{0}'(\mathbf{w}), 
\end{eqnarray}
which implies that $a.var(\hat \alpha^{\mbox{\tiny SC$'$}}_{0t} ) =   Q_{0}'(\bar{\mathbf{w}}^{\mbox{\tiny SC$'$}}) $, and $a.var (\hat \alpha^{\tiny DID}_{0t}  ) =   Q_{0}'(\bar{\mathbf{w}}^{\tiny DID} )$.  By definition of $\bar{\mathbf{w}}^{\mbox{\tiny SC$'$}}$, it must be that $Q_{0}'(\bar{\mathbf{w}}^{\mbox{\tiny SC$'$}})  \leq Q_{0}'(\bar{\mathbf{w}}^{\tiny DID} )$. 
\end{proof}

\subsubsection{Proposition \ref{Prop_specification} } \label{Proof_specification}

Consider the trivial identity 
\begin{eqnarray} \label{eq_chernozhukov}
0 = \left(\bar{\mathbf{w}} - \frac{1}{J} \mathbf{i} \right)' \left( \mathbf{y}_t - \omega \right) -  \left(\bar{\mathbf{w}} - \frac{1}{J} \mathbf{i} \right)' \left( \mathbf{y}_t - \omega \right),
\end{eqnarray}
where the demeaned SC weights converge to $\bar{\mathbf{w}}$, and $\omega = \mathbb{E}[\mathbf{y}_t - \mathbf{i}\delta_t]$. Note that these two vectors are well defined given the assumption that  $\lambda_t$ and $\epsilon_{jt}$ are stationary.

Following the notation from \cite{Chernozhukov}, we define $P^N_t = \left(\bar{\mathbf{w}} - \frac{1}{J} \mathbf{i} \right)' \left( \mathbf{y}_t - \omega \right)$ and $u_t = -\left(\bar{\mathbf{w}} - \frac{1}{J} \mathbf{i} \right)' \left( \mathbf{y}_t - \omega \right)$. Note that 
\begin{eqnarray}
u_t = -  \left(\bar{\mathbf{w}} - \frac{1}{J} \mathbf{i} \right)' \left( \boldsymbol{\mu} \lambda_t' +\boldsymbol{\epsilon}_t \right),
\end{eqnarray}
where we use the fact that $\bar{\mathbf{w}}'\mathbf{i} = \frac{1}{J} \mathbf{i}'\mathbf{i} = 1$ to eliminate $\delta_t$. Since $\lambda_t$ and $\boldsymbol{\epsilon}_t$ are weakly dependent stationary with mean zero, we have that $u_t$ is  weakly dependent stationary with mean zero.

Now consider 
\begin{eqnarray}
\hat P^N_t = \left( \widetilde{\mathbf{w}} -  \frac{1}{J} \mathbf{i} \right)' \left( \mathbf{y}_t - \frac{1}{T_0 + T_1} \sum_{\tau \in \mathcal{T}_0 \cup \mathcal{T}_1} \mathbf{y}_\tau \right) = -\hat u_t.
\end{eqnarray}

Note that 
\begin{eqnarray}
\hat P^N_t - P^N_t = \left( \widetilde{\mathbf{w}} -  \frac{1}{J} \mathbf{i} \right)' \left(  \frac{1}{T_0 + T_1} \sum_{\tau \in \mathcal{T}_0 \cup \mathcal{T}_1}\left( \boldsymbol{\mu} \lambda_\tau' +\boldsymbol{\epsilon}_\tau \right) \right) +  \left( \widetilde{\mathbf{w}} - \bar{\mathbf{w}} \right)' \left( \boldsymbol{\mu} \lambda_t' +\boldsymbol{\epsilon}_t \right),
\end{eqnarray}
where the first term on the RHS of the previous equation is $O_p(1) o_p(1)$, while the second one is $o_p(1)O_p(1)$.  Therefore, the model considered in equation \ref{eq_chernozhukov} satisfies all conditions for Theorem 1 from \cite{Chernozhukov}.

\subsection{Case with finite $T_0$} \label{finite_T}

We consider here the case with $T_0$ fixed.  For weights $\mathbf{w}^\ast  \in \widetilde \Phi$, note that:
\begin{eqnarray}  \label{Finite_T_original_reg}
y_{0t} =\mathbf{y}_t ' \mathbf{w}^\ast + \eta_{t}  \mbox{, for } t \leq 0 \mbox{, where }  \eta_{t}  = \epsilon_{0t} - \boldsymbol{\epsilon}_t' \mathbf{w}^\ast
\end{eqnarray}

Since $\sum_{j=1}^{J} w^\ast_j = 1$, we can write:
\begin{eqnarray}  \label{reg_finiteT}
\dot y_{0t} = \dot{\mathbf{y}}_t ' \dot{\mathbf{w}}^\ast  + \eta_{t}  
\end{eqnarray}
where $\dot y_{jt} = y_{jt} - y_{Jt}$, $\dot{\mathbf{y}}_t  = (\dot y_{1t},...,\dot y_{J-1,t}  )' $, and $\dot{\mathbf{w}}^\ast$ is the $J-1$ vector excluding the last entry of ${\mathbf{w}}^\ast$. The SC weights will be given by the OLS regression in \ref{reg_finiteT} with the non-negativity constraints, and with the constraint that the sum of the $J-1$ weights in $\widehat{\dot{\mathbf{w}}}^\ast $ must be smaller than 1. We ignore for now these constraints. Then we have that 
\begin{eqnarray}
\widehat{\dot{\mathbf{w}}}^\ast = \left( \sum_{t \in \mathcal{T}_0}\dot{\mathbf{y}}_t  \dot{\mathbf{y}}_t  ' \right)^{-1} \sum_{t \in \mathcal{T}_0} \dot{\mathbf{y}}_t   \dot y_{0t}.
\end{eqnarray}

We assume that $T_0$ is large enough so that $ \left( \sum_{t \in \mathcal{T}_0}\dot{\mathbf{y}}_t  \dot{\mathbf{y}}_t  ' \right)$ has full rank. Therefore:
\begin{eqnarray}
\mathbb{E}\left[\widehat{\dot{\mathbf{w}}}^\ast | \dot{\mathbf{y}}_{-T_0+1},...,\dot{\mathbf{y}}_0 \right] =\dot{\mathbf{w}}^\ast + \left( \sum_{t \in \mathcal{T}_0}\dot{\mathbf{y}}_t  \dot{\mathbf{y}}_t  ' \right)^{-1} \sum_{t \in \mathcal{T}_0}\dot{\mathbf{y}}_t \mathbb{E}[\eta_t |  \dot{\mathbf{y}}_{-T_0+1},...,\dot{\mathbf{y}}_0  ] 
\end{eqnarray}

By definition of $\eta_t$, we have that $\mathbb{E}[\eta_t |  \dot{\mathbf{y}}_{-T_0+1},...,\dot{\mathbf{y}}_0  ]  \neq 0$ for $t \leq 0$, which implies that $\widehat{\dot{\mathbf{w}}}^\ast $ is a biased estimator of ${\dot{\mathbf{w}}}^\ast $. Intuitively, the outcomes of the control units work as a proxy to the factor loadings of the treated unit. However, such proxy is imperfect, because the idiosyncratic shocks behave as a measurement error. 

If we consider the case without the non-negativity constraints, and assume that $\lambda_t$ and $\epsilon_{jt}$ are i.i.d. normal, then the conditional expectation function of $y_{0t}$ given $\mathbf{y}_t$ would be linear. As a consequence, the expected value of the SC weights would be exactly the same for any $T_0$, which, in turn, would be the same as the asymptotic value when $T_0 \rightarrow \infty$. If we relax the i.i.d. normality assumption and/or include the non-negativity constraints, then $\mathbb{E}\left[\widehat{\dot{\mathbf{w}}}^\ast | \dot{\mathbf{y}}_{-T_0+1},...,\dot{\mathbf{y}}_0 \right] $ would not be constant irrespectively  of $(\dot{\mathbf{y}}_{-T_0+1},...,\dot{\mathbf{y}}_0)$. However, the $\mathbb{E}\left[\widehat{\dot{\mathbf{w}}}^\ast \right]$ would be the integral of $\mathbb{E}\left[\widehat{\dot{\mathbf{w}}}^\ast | \dot{\mathbf{y}}_{-T_0+1},...,\dot{\mathbf{y}}_0 \right] $ over the distribution of  $(\dot{\mathbf{y}}_{-T_0+1},...,\dot{\mathbf{y}}_0)$. Therefore, we have no reason to believe that the distortion in the SC weights would be ameliorated if we consider a finite $T_0$ setting in comparison to the asymptotic distortion when $T_0 \rightarrow \infty$. 

Considering the non-negativity constraints would also affect the distribution of $\widehat{\dot{\mathbf{w}}}^\ast$ because, with finite $T_0$, there will be a positive probability that the solution to the unrestricted OLS problem will not satisfy the non-negativity constraints. However, this would not change the conclusion that  $\widehat{\dot{\mathbf{w}}}^\ast $ is a biased estimator of ${\dot{\mathbf{w}}}^\ast $.  In Section \ref{Sec_MC} we show MC simulations in which the distortion in the SC weights is aggravated when $T_0$ is small and we consider the non-negativity constraints. 

The larger bias of the SC weights when $T_0$ is smaller is discussed in detail for a particular set of linear factor models considered in the a previous version of our paper (see \cite{FP_old}). We present there a justification why we should expect (in that particular model) a larger bias for the SC weights when $T_0$ is finite.

\subsection{Setting with diverging common factors} \label{Appendix_diverging}

\subsubsection{Main Results with diverging common factors}

While the assumptions considered in Sections \ref{Bias_SC} and \ref{Sub_demeaned} allow for outcomes with divergent pre-treatment averages (which would be the case when we consider, for example, GDP or average wages), we restrict to settings in which such diverging common shocks affect all units in the same way. We now consider that case in which we may have diverging  common shocks that may have heterogeneous effects across unit.  We modify Assumption \ref{assumption_LFM} to include both common shocks that are non-diverging and diverging.

\begin{assumption_b}{\ref{assumption_LFM}$'$}[potential outcomes]
\normalfont

Potential outcomes are given by

\begin{eqnarray} \label{explosive_model}
\begin{cases} y_{jt}^N = c_j + \delta_t+  \gamma_t \theta_j + \lambda_t \mu_j + \epsilon_{jt}  \\ 
y_{jt}^I = \alpha_{jt} + y_{jt}^N. \end{cases}
\end{eqnarray}

\end{assumption_b}

We now separate the factor structure in two parts. One part, $\lambda_t \mu_j$,  that has the same properties as considered in  Sections \ref{Bias_SC} and \ref{Sub_demeaned}, and another one, $\gamma_t \theta_j $, which are ``diverging'', in the sense that pre-treatment averages of $\gamma_t$ diverge.  

\begin{assumption_b}{\ref{assumption_sample}$'$}[sampling] 

\normalfont We observe a realization of $\{ y_{0t} ,..., y_{Jt} \}_{t \in \mathcal{T}_0 \cup \mathcal{T}_1}$, where  $y_{jt} = d_{jt} y_{jt}^I  +  (1-d_{jt}) y_{jt}^N$, while $d_{jt}=1$ if $j=0$ and $t\in \mathcal{T}_1$, and zero otherwise. Potential outcomes are determined by equation (\ref{explosive_model}). We treat $\{ \mu_j\}_{j=0}^J$, $\{ \theta_j\}_{j=0}^J$, and $\{\gamma_t\}_{t \in   \mathcal{T}_0 \cup \mathcal{T}_1}$ as fixed, and $\{ \lambda_t \}_{t \in \mathcal{T}_0 \cup \mathcal{T}_1}$ and $\{ \epsilon_{jt} \}_{t \in \mathcal{T}_0 \cup \mathcal{T}_1}$ for $j=0,...,J$ as stochastic.

\end{assumption_b}

An important difference relative to the setting considered in Sections \ref{Bias_SC} and \ref{Sub_demeaned} is that we can consider a fixed sequence of $\gamma_t$. The idea is that, in this setting, we can find conditions in which the estimator is asymptotically unbiased even conditional on the realization of $\gamma_t$.\footnote{In contrast, the conditions for asymptotic unbiasedness considered in Sections \ref{Bias_SC} and \ref{Sub_demeaned} were valid over the distribution of $\lambda_t$.}  Since in this setting we expect  $\gamma_t$ to diverge as $T_0 \rightarrow \infty$, we have to consider the possibility that, for $\tau \in \mathcal{T}_1$,   $\gamma_\tau \rightarrow \infty$ when $T_0 \rightarrow \infty$.\footnote{We can think of that as a triangular array, where we fix a post-treatment periods $\tau$, and $\gamma_\tau$ potentially changes once we increase $T_0$. } The assumption below imposes restrictions on the sequence of $\gamma_t$ and on the other common and idiosyncratic shocks. Let $\tilde \gamma_t = [1 ~\gamma_t]$,  $\eta_j =  \lambda_t \mu_j +  \epsilon_{jt}$, and $\boldsymbol{\eta}_t = (\eta_{1t},...,\eta_{Jt})$, and consider $A = \mbox{diag}(T_0^{f_1},...,T_0^{f_{F_1}})$ and $\tilde A = \mbox{diag}(1,T_0^{f_1},...,T_0^{f_{F_1}})$ for constants $(f_1,...,f_{F_1}) \in \mathbb{R}^{F_1}_+$.

\begin{assumption_b}{\ref{assumptions_lambda}$'$}[Common and idiosyncratic shocks] 
\label{explosive_shocks}
\normalfont

$\exists (f_1,...,f_{F_1}) \in \mathbb{R}^{F_1}_+$ such that, \\ (i) $T_0^{-1}  \sum_{t \in \mathcal{T}_0} [\eta_{0t} ~ \boldsymbol{\eta}_t' ] \rightarrow  0$, (ii)  $T_0^{-1}  \sum_{t \in \mathcal{T}_0} [\eta_{0t} ~ \boldsymbol{\eta}_t' ]' [\eta_{0t} ~ \boldsymbol{\eta}_t' ] \rightarrow  \Sigma$ positive definite, (iii)  $T_0^{-1} \sum_{t \in \mathcal{T}_0}  \tilde A^{-1} \tilde \gamma_t' \tilde \gamma_t \tilde A^{-1} \rightarrow \Omega$ positive definite, (iv)  $T_0^{-1} \sum_{t \in \mathcal{T}_0}  \tilde A^{-1} \tilde \gamma_t'  \eta_{jt} \rightarrow 0 $ for all $j=0,...,J$, and (v) $A^{-1} \gamma_t = O(1)$.

\end{assumption_b}

Assumptions \ref{explosive_shocks}(i) and \ref{explosive_shocks}(ii) are equivalent to the assumptions we consider in Sections \ref{Bias_SC} and \ref{Sub_demeaned} for the ``non-diverging'' shocks.  Assumptions \ref{explosive_shocks}(iii), \ref{explosive_shocks}(iv) and \ref{explosive_shocks}(v)  determine the rates in which the components of $\gamma_t$ diverge.  Note that these assumptions would be satisfied if $\gamma_t$ is a polynomial trend. Moreover, we also show in a previous version of the paper that we can instead assume that $\gamma_t$ is a combination of ${I}(1)$ and polynomial trend factors (see  \cite{FP_old}).

We also consider an additional assumption on the factor loadings associated with the non-stationary common trends. Let $\Theta$ be the $J \times F_1$ matrix with information on the factor loadings $\theta_j$ of the controls.

\begin{assumption}[factor loadings] 
\label{non_stationary_assumption}
\normalfont

(i) $rank(\Theta) = F_1$, and (ii) $\exists ~ \textbf{w}^\ast \in W$ such that $\theta_0 = \Theta ' \mathbf{w}^\ast$, where $W$ is the set of possible weights given the constrains on the weights the researcher is willing to consider.

\end{assumption}

The first part of Assumption \ref{non_stationary_assumption} guarantees that the each diverging common shock generates enough independent variation on the outcomes of the controls. The second part of the assumption assumes existence of weights that reconstruct the factor loadings of  unit 0 associated with the non-stationary common trends. If this condition does not hold, then the asymptotic distribution of the SC estimators would trivially depend on the factor structure $\gamma_t \theta_i$.  Importantly,  we do \textit{not} need to assume existence of weights that satisfy Assumption \ref{non_stationary_assumption} and  also reconstruct $\mu_0$. Let $\Phi$ be the set of weights that reconstruct the factor loadings of both the diverging and non-diverging common shocks. 

We focus first on the demeaned SC estimator, and then we consider the  original SC estimator.

\begin{proposition} \label{I1_result}
\normalfont
Under Assumptions \ref{assumption_LFM}$'$, \ref{assumption_sample}$'$, \ref{assumption_exogeneity}, \ref{assumptions_lambda}$'$, and \ref{non_stationary_assumption}, for $\tau \in  \mathcal{T}_1$,
\begin{eqnarray}
\hat \alpha_{0 \tau}^{\mbox{\tiny SC$'$}}   \buildrel p \over \rightarrow  \alpha_{0\tau} + \left( \epsilon_{0\tau} - \mathbf{\bar w} '\boldsymbol{\epsilon}_\tau \right) + \lambda_\tau \left(\mu_0 - \boldsymbol{\mu}' \mathbf{ \bar w} \right) \mbox{ when } T_0 \rightarrow \infty
\end{eqnarray}
where $\mu_0 \neq  \boldsymbol{\mu}' \mathbf{ \bar w}$, unless  $ \sigma_\epsilon^2=0 $ or $\exists \textbf{w} \in \Phi |  \textbf{w} \in  \underset{{\textbf{w} \in W}}{\mbox{argmin}}  \left\{ \mathbf{w}'\mathbf{w}     \right\}$.

\end{proposition}

We present the proof in Appendix \ref{Tec_diverging}. Proposition \ref{I1_result} has two important implications. First, if  Assumption \ref{non_stationary_assumption} is valid, then the asymptotic distribution of the \emph{demeaned} SC estimator does not depend on the diverging common trends. The intuition of this result is the following. As $T_0 \rightarrow \infty$ minimizing the variance of a linear combination of the idiosyncratic shocks becomes irrelevant relative to the cost of failing to recover the factor loadings associated with the diverging common shocks. Therefore, we do not have the distortion on the SC weights we  find in Section \ref{Bias_SC} when we consider the diverging shocks. Interestingly, while $\mathbf{\widehat w}$ will generally be only $\sqrt{T_0}-$consistent when $\Phi_1 \equiv \{\mathbf{w} \in W | \theta_0 = \Theta' \mathbf{w}^\ast \}$ is not a singleton, we show in the proof  that there are linear combinations of $\mathbf{\widehat w}$ that will converge at a faster rate, implying that $\gamma_t ( \theta_0 - \sum_{j \neq 0}\hat  w_j \theta_j  ) \buildrel p \over \rightarrow 0$, despite the fact that $\gamma_t$ explodes when $T_0 \rightarrow \infty$. Therefore, such diverging common trends will not lead to asymptotic bias in the SC estimator. 

Second, the demeaned SC  estimator will be biased if there is correlation between treatment assignment and the non-diverging common factors $\lambda_t$. The intuition is that the demeaned SC weights will converge in probability to weights that minimize the asymptotic variance of  $u_t = y_{0t} - \mathbf{w}'\mathbf{y}_t= \lambda_t (\mu_0 - \boldsymbol{\mu}'\mathbf{w}) + (\epsilon_{0t} - \mathbf{w}'\boldsymbol{\epsilon}_t)$, restricting to the weights that satisfy Assumption \ref{non_stationary_assumption}.  Following the same arguments as in Proposition \ref{main_result}, $\mathbf{\widehat w}$ will not eliminate these non-diverging common factors, unless we have that $ \sigma_\epsilon^2=0 $ or it coincides that there is a $\textbf{w} \in \Phi$ that also minimizes the linear combination of idiosyncratic shocks.

The result that the  asymptotic distribution of the SC estimator does not depend on the non-stationary common  trends depends crucially on  Assumption \ref{non_stationary_assumption}. If there were no linear combination of the control units that reconstruct the factor loadings of the treated unit associated to the diverging common trends, then the asymptotic distribution of the SC estimator would trivially depend on these common trends, which might lead to bias in the SC estimator if treatment assignment is correlated with such diverging trends.

Proposition \ref{I1_result}  remains valid when  we relax the adding-up and/or the non-negativity constraints, with minor variations in the conditions for unbiasedness. However, these results are not valid when we consider the no-intercept constraint, as the original SC estimator does. When the intercept is not included, it remains true that  $\mathbf{\widehat w} \buildrel p \over \rightarrow \mathbf{\bar w} \in \Phi_1$.  However,  in this case, the weights will not converge fast enough to compensate the fact that $\gamma_t$ explodes, implying that the result from Proposition \ref{I1_result} that the asymptotic distribution of the estimator does not depend on the diverging common factor does not hold if we consider the estimator with no intercept. We present a counter-example in Appendix \ref{Tec_diverging}.

\subsubsection{Technical results with  diverging common factors} \label{Tec_diverging}

\textbf{Proof of Proposition \ref{I1_result} without constraints}

We show this result  for the case without the adding-up, non-negativity, and no intercept constraints. In this case, the time fixed effects $\delta_t$ may enter either in the $\gamma_t$ or in the $\lambda_t$ vectors.  We then  extend these results for the cases with the adding-up and/or non-negativity constraints. After that,  we show a counterexample in which  this result is not valid when we use the no intercept constraint.

First, let $\Theta_a^b$ contain the rows $a$ to $b$ of matrix $\Theta$. If we set $a=0$, then the first row of  $\Theta_a^b$ is given by $\theta_0'$. Since $rank(\Theta) = F_1$, we can assume, without loss of generality, that $rank(\Theta_{J-F_1+1}^J)$ (that is, the last $F_1$ control units have $\theta_j$ that form a basis of $\mathbb{R}^{F_1}$. 
Therefore, we have
\begin{eqnarray}\nonumber
\begin{bmatrix}   y_{0,t} \\ y_{1,t} \\ \vdots \\ y_{J-F_1,t}  \end{bmatrix} &=& \begin{bmatrix}   c_{0} \\ c_{1} \\ \vdots \\ c_{J-F_1}\end{bmatrix} -  \Theta_{0}^{J-F_1} (\Theta_{J-F_1+1}^J)^{-1} \begin{bmatrix}   c_{J-F_1+1} \\ \\ \vdots \\ c_{J}\end{bmatrix}   + \Theta_{0}^{J-F_1} (\Theta_{J-F_1+1}^J)^{-1}\begin{bmatrix}   y_{J-F_1+1,t} \\ \vdots \\ y_{J,t}   \end{bmatrix}  \\
&& + \begin{bmatrix}   \eta_{0,t} \\ \eta_{1,t} \\ \vdots \\ \eta_{J-F_1,t}\end{bmatrix} -  \Theta_{0}^{J-F_1} (\Theta_{J-F_1+1}^J)^{-1} \begin{bmatrix}   \eta_{J-F_1+1,t} \\ \\ \vdots \\ \eta_{J,t}\end{bmatrix},
\end{eqnarray}
which is similar to the triangular representation from \cite{Phillips1991} for cointegrating relations.

We re-write this equation as
\begin{eqnarray} \label{eq_structure}
\begin{bmatrix}   y_{0,t} \\ y_{1,t} \\ \vdots \\ y_{J-F_1,t}  \end{bmatrix} &=& \begin{bmatrix} \bar  c_{0} \\  \bar c_{1} \\ \vdots \\ \bar c_{J-F_1}\end{bmatrix} + \Theta_{0}^{J-F_1} \begin{bmatrix}   \tilde y_{J-F_1+1,t} \\ \vdots \\ \tilde y_{J,t}   \end{bmatrix}   + \begin{bmatrix}   \bar \eta_{0,t} \\ \bar  \eta_{1,t} \\ \vdots \\ \bar \eta_{J-F_1,t}\end{bmatrix}.
\end{eqnarray}

Now define $\beta \in \mathbb{R}^{F_1}$ such that $u_t = \bar \eta_{0,t} - [\bar \eta_{1,t} ~ \hdots ~ \bar \eta_{J-F_1}]'\beta \rightarrow_p 0 $, and consider the OLS regression of $\bar \eta_{0,t}$ on  $\bar{\boldsymbol{\eta}}_t \equiv (\bar \eta_{1,t}, \hdots,  \bar \eta_{J-F_1})$, a constant,  and $ \tilde{\mathbf{y}}_t \equiv (  \tilde y_{J-F_1+1,t},\hdots,\tilde y_{J,t})$. The OLS estimators ($\hat \beta$, $\hat \kappa$. and $\hat \phi$) are given by

\begin{eqnarray}
\begin{bmatrix}    \hat \beta - \beta \\ \hat \kappa \\  A \hat \phi  \end{bmatrix} \nonumber
 &=&  \begin{bmatrix}    T_0^{-1}   \sum_{t \in \mathcal{T}_0} \bar{\boldsymbol{\eta}}_t \bar{\boldsymbol{\eta}}_t '   &  T_0^{-1}   \sum_{t \in \mathcal{T}_0} \bar{\boldsymbol{\eta}}_t   & T_0^{-1}   \sum_{t \in \mathcal{T}_0}\bar{\boldsymbol{\eta}}_t \left( A^{-1} \tilde{\mathbf{y}}_t \right)  '   \\  
 T_0^{-1}   \sum_{t \in \mathcal{T}_0} \bar{\boldsymbol{\eta}}_t '    & 1  &  T_0^{-1}   \sum_{t \in \mathcal{T}_0} \left( A^{-1} \tilde{\mathbf{y}}_t \right) '    \\ 
   T_0^{-1}   \sum_{t \in \mathcal{T}_0}  \left( A^{-1} \tilde{\mathbf{y}}_t \right) \bar{\boldsymbol{\eta}}_t '  &  T_0^{-1}   \sum_{t \in \mathcal{T}_0} \left( A^{-1} \tilde{\mathbf{y}}_t \right)  & T_0^{-1}   \sum_{t \in \mathcal{T}_0} \left( A^{-1} \tilde{\mathbf{y}}_t \right) \left( A^{-1} \tilde{\mathbf{y}}_t \right) '
   \end{bmatrix}^{-1} \times \\
   && \times  \begin{bmatrix}  T_0^{-1}   \sum_{t \in \mathcal{T}_0} \bar{\boldsymbol{\eta}}_t  u_t  \\  T_0^{-1}   \sum_{t \in \mathcal{T}_0}u_t  \\  T_0^{-1}   \sum_{t \in \mathcal{T}_0} \tilde{\mathbf{y}}_t  u_t    \end{bmatrix}.
\end{eqnarray}

From Assumption \ref{explosive_shocks}, we have that 
\begin{eqnarray}
\begin{bmatrix}    T_0^{-1}   \sum_{t \in \mathcal{T}_0} \bar{\boldsymbol{\eta}}_t \bar{\boldsymbol{\eta}}_t '   &  T_0^{-1}   \sum_{t \in \mathcal{T}_0} \bar{\boldsymbol{\eta}}_t   & T_0^{-1}   \sum_{t \in \mathcal{T}_0}\bar{\boldsymbol{\eta}}_t \left( A^{-1} \tilde{\mathbf{y}}_t \right)  '   \\  
 T_0^{-1}   \sum_{t \in \mathcal{T}_0} \bar{\boldsymbol{\eta}}_t '    & 1  &  T_0^{-1}   \sum_{t \in \mathcal{T}_0} \left( A^{-1} \tilde{\mathbf{y}}_t \right) '    \\ 
   T_0^{-1}   \sum_{t \in \mathcal{T}_0}  \left( A^{-1} \tilde{\mathbf{y}}_t \right) \bar{\boldsymbol{\eta}}_t '  &  T_0^{-1}   \sum_{t \in \mathcal{T}_0} \left( A^{-1} \tilde{\mathbf{y}}_t \right)  & T_0^{-1}   \sum_{t \in \mathcal{T}_0} \left( A^{-1} \tilde{\mathbf{y}}_t \right) \left( A^{-1} \tilde{\mathbf{y}}_t \right) '
   \end{bmatrix}   \rightarrow_p \begin{bmatrix} \Sigma & 0 \\ 0 & \Omega   \end{bmatrix}, 
\end{eqnarray}
which is positive definite, and we also have that 
\begin{eqnarray}
\begin{bmatrix}  T_0^{-1}   \sum_{t \in \mathcal{T}_0} \bar{\boldsymbol{\eta}}_t  u_t  \\  T_0^{-1}   \sum_{t \in \mathcal{T}_0}u_t  \\  T_0^{-1}   \sum_{t \in \mathcal{T}_0} \tilde{\mathbf{y}}_t  u_t    \end{bmatrix} \rightarrow_p 0.
\end{eqnarray}

Therefore, 
\begin{eqnarray}
\begin{bmatrix}    \hat \beta - \beta \\ \hat \kappa \\  A \hat \phi  \end{bmatrix} \rightarrow_p 0.
\end{eqnarray}

Now note that, from equation (\ref{eq_structure}), we have that
\begin{eqnarray}
y_{0,t} &=& \begin{bmatrix}   1 & -\hat \beta' \end{bmatrix}  \begin{bmatrix} \bar  c_{0} \\  \bar c_{1} \\ \vdots \\ \bar c_{J-F_1}\end{bmatrix}  + \hat \kappa + \hat \beta'  \begin{bmatrix}    y_{1,t} \\ \vdots \\ y_{J-F_1,t}  \end{bmatrix} + \\ 
&& + \left( \begin{bmatrix}   1 & -\hat \beta' \end{bmatrix} \Theta_0^{J-F_1} (\Theta_{J-F_1+1}^J)^{-1}  + \hat \phi  (\Theta_{J-F_1+1}^J)^{-1} \right)\begin{bmatrix}   y_{J-F_1+1,t} \\ \vdots \\ y_{J,t}   \end{bmatrix}  + \hat u_t,
\end{eqnarray}
which implies that an OLS regression of $y_{0,t}$ on a constant, $ (y_{1,t}, \hdots,  y_{J-F_1})$, and $  (   y_{J-F_1+1,t},\hdots, y_{J,t})$ yields estimators $\hat c = \begin{bmatrix}   1 & -\hat \beta' \end{bmatrix}  \begin{bmatrix} \bar  c_{0} &  \bar c_{1} & \hdots & \bar c_{J-F_1}\end{bmatrix}' +  \hat \kappa $, $\hat \beta$, and $\left( \begin{bmatrix}   1 & -\hat \beta' \end{bmatrix} \Theta_0^{J-F_1} (\Theta_{J-F_1+1}^J)^{-1}  + \hat \phi  (\Theta_{J-F_1+1}^J)^{-1} \right)$.

We are interested in the limiting distribution of $\hat \alpha_{0\tau}^{\mbox{\tiny SC$'$}}$, for $\tau \in \mathcal{T}_1$:
\begin{eqnarray} \nonumber
\hat \alpha_{0\tau}^{\mbox{\tiny SC$'$}} &=& y_{0\tau} - \mathbf{y}_\tau' \widehat{\textbf{w}}^{\mbox{\tiny SC$'$}} =   \alpha_{0\tau}  + \lambda_\tau \left(\mu_0 - \boldsymbol{\mu}' \widehat{\textbf{w}}^{\mbox{\tiny SC$'$}} \right) +\gamma_\tau \left(\theta_0 - \Theta'\widehat{\textbf{w}}^{\mbox{\tiny SC$'$}}  \right) + \left( \epsilon_{0\tau} - \boldsymbol{\epsilon}_{\tau}'\widehat{\textbf{w}}^{\mbox{\tiny SC$'$}} \right) \\
&& + c_0 - [c_1 ~ \hdots ~ c_J]\widehat{\textbf{w}}^{\mbox{\tiny SC$'$}} -  \hat c.
\end{eqnarray}

With some algebra, we have that 
\begin{eqnarray} 
\gamma_\tau \left(\theta_0 - \Theta'\widehat{\textbf{w}}^{\mbox{\tiny SC$'$}}  \right) = \gamma_\tau \hat \phi = ( \gamma_\tau A^{-1})(A\hat \phi) = o_p(1).
\end{eqnarray}

Likewise, we have that
\begin{eqnarray} 
c_0 - [c_1 ~ \hdots ~ c_J]\widehat{\textbf{w}}^{\mbox{\tiny SC$'$}} -  \hat c = \hat \kappa = o_p(1),
\end{eqnarray}
implying that 
\begin{eqnarray} 
\hat \alpha_{0\tau}^{\mbox{\tiny SC$'$}} \rightarrow_p   \alpha_{0\tau}  + \lambda_\tau \left(\mu_0 - \boldsymbol{\mu}' {\mathbf{\bar w}} \right) + \left( \epsilon_{0\tau} - \boldsymbol{\epsilon}_{\tau}' {\mathbf{\bar w}} \right).
\end{eqnarray}

Finally, by definition of $u_t$, the OLS estimator converges to weights that minimize $\mbox{plim}[(y_{0t} - \mathbf{y}_t ' \mathbf{w})^2]$ subject to $\mathbf{w} \in \Phi_1$. Therefore, the proof that $\mathbf{\widehat w} \buildrel p \over \rightarrow \bar{\mathbf{w}} \notin \Phi$  is essentially the same as the proof of Proposition \ref{main_result}. 

\

\textbf{Proof of Proposition \ref{I1_result} with adding-up and non-negativity constraints}

To show that this result is also valid for the case with adding-up constraint we just have to  consider the OLS regression of $y_{0t} - y_{1t}$ on a constant and $y_{2t}-y_{1t},...,y_{Jt}-y_{1t}$. Under Assumption \ref{non_stationary_assumption}, this transformed model is also cointegrated, so we can  apply our previous result.

We now consider the case with the non-negative constraints. We prove the case $W=\{\textbf{w} \in \mathbb{R}^J ~ | ~ w_j \geq 0\}$. Including an adding-up constraint then follows directly from a change in variables as we did for the case without non-negative constraints.

We first show that $\mathbf{\widehat w} \buildrel p \over \rightarrow \bar{\mathbf{w}}$ where $\bar{\mathbf{w}}$ minimizes $\mathbb{E}[u_t^2]$ subject to $\mathbf{w} \in \Phi_1 \cap W$. Suppose that $\bar{\mathbf{w}} \in int(W)$. This implies that $\bar{\mathbf{w}} \in int(\Phi_1 \cap W)$ relative to $\Phi_1$.  By convexity of $E[u_t^2]$,  $\bar{\mathbf{w}} $ also minimizes $E[u_t^2]$ subject to $\Phi_1$. We know that OLS without the non-negativity constraints converges in probability to $\bar{\mathbf{w}}$. Let $\widehat{\textbf{w}}_{u}$ be the OLS estimator without the non-negativity constraints and $\widehat{\textbf{w}}_{r}$ be the OLS estimator with the non-negativity constraint. Since $\bar{\mathbf{w}} \in {int}( W)$, then it must be that, for all $\epsilon>0$, $||\widehat{\textbf{w}}_{u} - \bar{\mathbf{w}}||<\epsilon$ with probability approaching to 1 (w.p.a.1). Since $\widehat{\textbf{w}}_{u} =  \widehat{\textbf{w}}_{r} $ when $\widehat{\textbf{w}}_{u} \in int(W)$ (due to convexity of the OLS objective function), these two estimators are asymptotically equivalent.  

Consider now the case in which $\mathbf{\bar w}$ is on the boundary of $W$.  This means that $\bar w_j=0$ for at least one $j$. Let $A=\{ j | w_j^\ast=0 \}$. Note first that $\mathbf{\bar w}$ also minimizes $E[u^2_t]$ subject to $\textbf{w} \in \Phi_1 \cap \{ \textbf{w} |  w_j=0 ~ \forall j \in A \}$. That is, if we impose the restriction $w_j=0$ for all $j$ such that $\bar w_j=0$, then we would have the same minimizer, even if we ignore the other non-negative constraints. Suppose there is an $\mathbf{\tilde w} \neq \mathbf{\bar w}$ that minimizes $E[u^2_t]$ subject to $\textbf{w} \in \Phi_1 \cap \{ \textbf{w} |  w_j=0 ~ \forall j \in A \}$.  By strict convexity of the objective function and the fact that $\mathbf{\bar w}$ is in the interior of $\Phi \cap W \cap  \{ \textbf{w} |  w_j=0 ~ \forall j \in A \} $ relative to $\Phi_1 \cap  \{ \textbf{w} |  w_j=0 ~ \forall j \in A \} $, there must be $\textbf{w}' \in \Phi_1 \cap W \cap  \{ \textbf{w} |  w_j=0 ~ \forall j \in A \}  \subset  \Phi_1 \cap W$ that attains a lower value in the objective function than $\mathbf{\bar w}$. However, this contradicts the fact that $\mathbf{\bar w} \in \Phi_1 \cap W $ is the minimum. 

Now let $\widehat{\textbf{w}}' $ be the OLS estimator subject to  $ \{ \textbf{w} |  w_j=0 ~ \forall j \in A \}$. We have that $\widehat{\textbf{w}}' $ is consistent for $\mathbf{\bar w}$. Now we show that  $\widehat{\textbf{w}}' $ is asymptotically equivalent to $\widehat{\textbf{w}}'' $, the OLS estimator subject to   $ \{ \textbf{w} |  w_j \geq 0 ~ \forall j \in A \}$. We prove the case in which $A = \{ j \}$ (there is only one restriction that binds). The general case follows by induction. Suppose  these two estimators are not asymptotically equivalent. Then there is $\epsilon>0$ such that $Lim Pr(|\widehat{\textbf{w}}' -\widehat{\textbf{w}}'' | >\epsilon ) \neq 0$. There are two possible cases.

First, suppose  that $\mbox{Lim} Pr \left( |  \widehat w''_j | > \epsilon'   \right)=0$ for all $\epsilon'>0$ (that is, the OLS subject to  $ \{ \textbf{w} |  w_j \geq 0 ~ \forall j \in A \}$ converges in probability to $\mathbf{\bar w}$ such that $\bar w_j=0$). However, since the two estimators are not asymptotically equivalent, for all $T_0'$, we can always find a $T_0>T_0'$ such that, with positive probability,  $|\widehat{\textbf{w}}' -\widehat{\textbf{w}}'' | >\epsilon$. Since $\{ \textbf{w} |  w_j=0 ~ \forall j \in A \} \subset   \{ \textbf{w} |  w_j \geq 0 ~ \forall j \in A \}$ and $\widehat{\textbf{w}}' \neq \widehat{\textbf{w}}'' $,  then $Q_{T_0}( \widehat{\textbf{w}}'') < Q_{T_0}( \widehat{\textbf{w}}')$, where $Q_{T_0}()$ is the OLS  objective function. Now using the continuity of the OLS objective function and the fact that $\widehat w_j''$ converges in probability to zero, we can always find $T_0'$ such that there will be a positive probability that  $Q_{T_0}( \widehat{\textbf{w}}'' - e_j \hat w_j'') < Q_{T_0}( \widehat{\textbf{w}}')$. Since $\widehat{\textbf{w}}'' - e_j \hat w_j'' \in \{ \textbf{w} |  w_j=0 ~ \forall j \in A \} $, this contradicts $\widehat{\textbf{w}}'$ being OLS subject to $  \{ \textbf{w} |  w_j = 0 ~ \forall j \in A \}$.

Alternatively, suppose that there exists $\epsilon'>0$ such that  $\mbox{Lim} Pr \left( |  \widehat w''_j | > \epsilon'   \right) \neq 0$. This means that, for all $T_0'$, we can find $T_0 > T_0'$ such that  there is a positive probability that the solution to OLS on $ \{ \textbf{w} |  w_j \geq 0 ~ \forall j \in A \}$ is in an interior point $\widehat{\textbf{w}}''$ with $\hat w_j'' > \epsilon'>0$. By convexity of $Q_{T_0}()$, this would imply that $\widehat{\textbf{w}}''$ is also the solution to the OLS without any restriction. However, this contradicts the fact that OLS without non-negativity restriction is consistent (see proof of Proposition \ref{I1_result}).

Finally, we show that  $\widehat{\textbf{w}}''$ and  $\widehat{\textbf{w}}_r$ are asymptotically equivalent. Note that $\mathbf{\bar w}$ is in the interior of $W$ relative to  $ \{ \textbf{w} |  w_j \geq 0 ~ \forall j \in A \}$. Therefore, w.p.a.1,  $\widehat{\textbf{w}}'' \in W$, which implies that $\widehat{\textbf{w}}''=\widehat{\textbf{w}}_r$.

We still need to show that linear combinations of $\mathbf{\widehat w}_r$ converge fast enough to reconstruct the factor loadings of the treated unit associated with the non-stationary common factors, so that $\gamma_t ( \theta_0 - \sum_{j \neq 0}\hat  w^r_j \theta_j  ) \buildrel p \over \rightarrow 0$. Let $Q_{T_0}()$ be the OLS objective function, and let $\mathcal{\widetilde W} = \{\mathbf{\widetilde w}_1,...,\mathbf{\widetilde w}_{2^J} \}$ be the set of all possible OLS estimators when we consider some of the non-negative constraints as equality and ignore the other ones. Let $\mathcal{\widetilde W}' \subset \mathcal{\widetilde W}$ be the set of estimators in $\mathcal{\widetilde W} $ such that all non-negative constraints are satisfied. Then we know that $\mathbf{\widehat w}_r = argmin_{\mathbf{w} \in \mathcal{\widetilde W}'} Q_{T_0}(\mathbf{w})$.

Suppose first that, for each of the $2^J$ combinations of restrictions, there is at least one $\mathbf{w} \in \Phi_1$ that satisfy these restrictions. In this case, we know from the first part of the proof that  $\gamma_t \left( \theta_0 - \sum_{j \neq 0}\widetilde  w_j^h \theta_j  \right) \buildrel p \over \rightarrow 0$ for all $h=1,...,2^J$, where $\mathbf{\widetilde w}_h = (\widetilde  w_1^h,...,\widetilde  w_J^h)'$. Moreover, since $\mathcal{\widetilde W} $ is finite, then this convergence is uniform in $\mathcal{\widetilde W} $. Therefore, it must be that $\gamma_t ( \theta_0 - \sum_{j \neq 0}\hat  w^r_j \theta_j  ) \buildrel p \over \rightarrow 0$.  Suppose now that for the combination of restrictions considered for   $\mathbf{\widetilde w}_h$, with $h \in \{1,...,2^J \}$, there is no $\mathbf{w} \in \Phi_1$ that satisfies these restrictions. Since the parameter space with this combination of restrictions is closed, then $\exists \eta>0$ such that $||\theta_0 - \sum_{j \neq 0} w_j \theta_j ||>\eta$ for all $\mathbf{w}$ that satisfy this combinations of restrictions.\footnote{Otherwise, there would be $\mathbf{w} \in \Phi_1$ that satisfies this combination of restrictions.} Therefore, $Q_{T_0}(\mathbf{\widetilde w}_h)$ diverge when $T_0 \rightarrow \infty$, implying that, w.p.a.1, $\mathbf{\widehat w}_r \neq \mathbf{\widetilde w}_h$.

\

\textbf{Example with no intercept} 
 
We consider now a very simple example to show that it is not possible to guarantee that $\gamma_t \left( \theta_0 - \sum_{j \neq 0}\hat  w_j \theta_j  \right) \buildrel p \over \rightarrow 0$ if we do not include the intercept.  Consider the case in which there are only one treated and one control unit, and $y_{0t} = \mu_0 + t + u_{0t}$ while $y_{1t} = \mu_1 + t + u_{1t}$. We consider a regression of $y_{0t}$ on $y_{1t}$ without the intercept. Note that $y_{0t} = (\mu_0 - \mu_1) + y_{1t} + u_{0t } - u_{1t} = \mu + y_{1t}  + u_t$. Then we have that:
\begin{eqnarray}
\hat \beta = \frac{\sum_{t=1}^{T_0}y_{1t}y_{0t}}{\sum_{t=1}^{T_0}y_{1t}^2} = 1 + \frac{\sum_{t=1}^{T_0} ( \mu \mu_1 + \mu t + \mu u_{1t} + \mu_1 u_t + t u_t + u_t u_{1t}) }{\sum_{t=1}^{T_0} (t^2 + \mu_1^2 + u_{1t}^2 + \mbox{``cross terms''})}
\end{eqnarray}
which implies that:
\begin{eqnarray}
T (\hat\beta-1) =  \frac{\frac{1}{T^2} \sum_{t=1}^{T_0} ( \mu \mu_1 + \mu t + \mu u_{1t} + \mu_1 u_t + t u_t + u_t u_{1t}) }{ \frac{1}{T^3} \sum_{t=1}^{T_0} (t^2 + \mu_1^2 + u_{1t}^2 + \mbox{``cross terms''})} \buildrel p \over \rightarrow  \frac{\frac{1}{2} \mu}{\frac{1}{3}}
\end{eqnarray}

Therefore, while $\hat \beta \buildrel p \over \rightarrow 1$, it does not converge fast enough so that $T (\hat\beta-1) \buildrel p \over \rightarrow 0$, except when $\mu_0=\mu_1$.

\subsection{Example: SC Estimator vs DID Estimator} \label{example}

We provide an example in which the asymptotic bias of the SC estimator can be higher than  the asymptotic bias of the DID estimator. Assume we have 1 treated and 4 control units in a model with 2 common factors. For simplicity, assume that there is no additive fixed effects and that $\mathbb{E}[\lambda_t]=0$. We have that the factor loadings are given by:
\begin{eqnarray}
\mu_0 = \left( \begin{array}{c} 1 \\ 1 \end{array} \right) \mbox{, } \mu_2 = \left( \begin{array}{c} 0.5 \\ 1 \end{array} \right) \mbox{, } \mu_3 = \left( \begin{array}{c} 1.5 \\ 1 \end{array} \right) \mbox{, } \mu_4 = \left( \begin{array}{c} 0.5 \\ 0 \end{array} \right) \mbox{, } \mu_5 = \left( \begin{array}{c} 1.5 \\ 1 \end{array} \right)
\end{eqnarray}

Note that any linear combination $0.5 \mu_2 + w_1^3 \mu_3 + w_1^5 \mu_5$ with $w_1^3+w_1^5 = 0.5$ recovers $\mu_0$. Note also that DID equal weights would set the first factor loading to 1, which is equal to $\mu_0^1$, but the second factor loading would be equal to $0.75 \neq \mu_0^2$. We want to show that the SC weights would improve the construction of the second factor loading but it will distort the combination for the first factor loading. If we set $\sigma_\epsilon^2=\mathbb{E}[(\lambda_t^1)^2]=\mathbb{E}[(\lambda_t^2)^2]=1$, then the factor loadings of the SC unit would be given by $(1.038,0.8458)$. Therefore, there is small loss in the construction of the first factor loading and a gain in the construction of the second factor loading. Therefore, if selection into treatment is correlated with the common shock $\lambda_t^1$, then the SC estimator would be more asymptotically biased than the DID estimator.

\subsection{Alternatives specifications and alternative estimators}  \label{A_alternatives}

\subsubsection{Average of pre-intervention outcome as economic predictor}

We consider now another very common specification in SC applications, which is to use the average pre-treatment outcome as the economic predictor. Note that if one uses only the average pre-treatment outcome as the economic predictor then the choice of matrix $V$ would be irrelevant. In this case, the minimization problem would be given by:
\begin{eqnarray} \nonumber
\{ {\hat w}_j \}_{j \neq 0} &=& \mbox{argmin}_{w \in \Delta^{J-1}}  \left[\frac{1}{T_0} \sum_{t \in \mathcal{T}_0} \left(  y_{0t} - \sum_{j \neq 0} w_j y_{jt}  \right) \right]^2  \\
&=& \mbox{argmin}_{w \in \Delta^{J-1}}  \left[ \frac{1}{T_0} \sum_{t \in \mathcal{T}_0} \left(  \epsilon_{0t} - \sum_{j \neq 0} w_j \epsilon_{jt} + \lambda_t \left(\mu_0 - \sum_{j \neq 0} w_j \mu_j \right) + c_0 - \sum_{j\neq 0 } w_j c_j  \right)\right]^2.
\end{eqnarray}

Therefore, under Assumptions \ref{assumption_sample}, \ref{assumption_exogeneity} and  \ref{assumptions_lambda},  the objective function converges in probability to:
\begin{eqnarray}
\Gamma( \textbf{w} ) =   \left(c_0 - \sum_{j \neq 0} w_j c_j \right)^2  
\end{eqnarray}

Therefore, if there are weights that reconstruct the unit fixed effects without reconstructing the other factor loadings of the treated unit, then  there is no guarantee that the SC control method will choose weights that are close to the correct ones. This result is consistent with the MC simulations by \cite{FPP}, who show that this specification performs particularly bad in allocating the weights correctly. 

\subsubsection{Adding other covariates as predictors} \label{theta}

Most SC applications that use  the average pre-intervention outcome value as economic predictor also consider other time invariant covariates as economic predictors. Let $Z_i$ be a $(R \times 1)$ vector of observed covariates (not affected by the intervention). Assumption \ref{assumption_LFM} changes to:
\begin{eqnarray} 
\begin{cases} y_{it}^N = \delta_t + c_i + \theta_t Z_i+ \lambda_t \mu_i + \epsilon_{it}  \\ 
y_{it}^I = \alpha_{it} + y_{it}^N \end{cases}
\end{eqnarray}

We redefine the set $\Phi = \{ \textbf{w} \in \Delta^{J-1} ~ | ~ c_0 = \sum_{j\neq 0} c_j w_j, \mu_0 = \sum_{j \neq 0} {w_j} \mu_j \mbox{, } Z_0 = \sum_{j \neq 0} {w_j} Z_j  \}$. Let $X_1$ be an $((R +1) \times 1)$ vector that contains the average pre-intervention outcome  and all covariates for unit 1, while $X_0$ is a $((R +1) \times J)$ matrix that contains    the same information for the control units. For a given $V$, the first step of the nested optimization problem suggested in \cite{Abadie2010} would be given by:
\begin{eqnarray}  \label{1st_problem}
\widehat{\textbf{w}}(V) \in \mbox{argmin}_{\textbf{w} \in \Delta^{J-1}} || X_1 - X_0 \textbf{w}  ||_V.
\end{eqnarray}

Considering again the assumptions from Section \ref{Bias_SC}, the  objective function of this minimization problem converges to $|| \bar X_1 - \bar X_0 \textbf{w}  ||_V$, where:
\begin{eqnarray}
 \bar X_1 - \bar X_0 \textbf{w} =  \left[  \begin{array}{c} 
\bar \theta \left(Z_0 - \sum_{j \neq 0} w_j Z_j \right) +  \left(c_0 - \sum_{j \neq 0} w_jc_j \right)  \\
 \left(Z_0^1 - \sum_{j \neq 0} w_j Z_j^1 \right)  \\
\vdots \\
  \left(Z_0^R - \sum_{j \neq 0} w_j Z_j^R \right) \\
   \end{array}      \right],
\end{eqnarray}
  where we assume $\frac{1}{T_0} \sum_{t \in \mathcal{T}_0} \theta_t \rightarrow_p \bar \theta$.  Therefore, there is no guarantee that an estimator based on this minimization problem would converge to weights in $\Phi$ for any given matrix $V$, even if $\Phi \neq \varnothing$.  

The second step in the nested optimization problem is to choose $V$ such that $\widehat{\textbf{w}}(V) $ minimizes the pre-intervention prediction error. Note that this problem is essentially given by:
\begin{eqnarray} 
\widehat{\textbf{w}} &=& \mbox{argmin}_{w \in \widetilde W}  \left[\frac{1}{T_0} \sum_{t \in \mathcal{T}_0} \left(  y_{0t} - \sum_{j \neq 0} w_j y_{jt}  \right) \right]^2  
\end{eqnarray}
where $\widetilde W  \subseteq \Delta^{J-1}$ is the set of $\textbf{w}$  such that ${\textbf{w}}$ is the  solution to problem \ref{1st_problem} for some positive semidefinite matrix $V$. Similarly to the SC estimator that includes  all pre-treatment outcomes, there is no guarantee that this minimization problem will choose weights in $\Phi$, even when $T_0 \rightarrow \infty$.  
Therefore, it is not possible to guarantee that this SC estimator would be asymptotically unbiased. MC simulation presented by  \cite{FPP} confirm that this SC specification systematically misallocates more weight than alternatives that use a large number of pre-treatment outcome lags as predictors.

\subsubsection{Relaxing constraints on the weights \& other estimators} \label{relaxing_constraints}

Our main result that the original and the demeaned SC estimators are generally asymptotically biased if there are unobserved time-varying confounders (Propositions \ref{main_result} and \ref{Prop_demeaned}) still applies if we also relax the non-negative and the adding-up constraints, which essentially leads to the panel data approach suggested by \cite{Hsiao}, and further explored by \cite{Li}.\footnote{In this case,  since we do not constraint the weights to sum 1, we need to adjust Assumption \ref{assumptions_lambda} so that it also includes convergence of the pre-treatment averages of the first and second moments of $\delta_t$.} Our conditions for unbiasedness of the SC estimator also apply to the estimators proposed by  \cite{Carvalho2015} and  \cite{Carvalho2016b} when $J$ is fixed. 

These papers rely on assumptions that essentially imply no selection on unobservables to derive consistency results, which reconciles our results with theirs.   \cite{Hsiao} and  \cite{Li}  implicitly rely on stability in the linear projection of the potential outcomes of the treated unit on the outcomes of the control units, before and after the intervention, to show that their proposed estimators are  unbiasedness and consistent. See, for example, equation A.4 from \cite{Li}. For simplicity, consider that $\lambda_t \mu_i$ includes the fixed effects $c_i$ and $\delta_t$. Then the linear projection of $y_{0t}^N$ given $\mathbf{y}_t$ for any given $t$ is given by $\delta_1(t) + \mathbf{y}_t ' \delta(t) $, where
\begin{eqnarray} \label{linear_projection}
\begin{cases}
\delta(t) = \left[ \boldsymbol{\mu} var(\lambda_t ) \boldsymbol{\mu}'  \right]^{-1} \boldsymbol{\mu} var(\lambda_t ) \mu_0 \mbox{, and} \\
\delta_1(t) = \mathbb{E}[\lambda_t ] (\mu_0 -   \boldsymbol{\mu}' \delta(t)).
\end{cases}
\end{eqnarray}

Therefore, in general, we will only have $(\delta_1(t),\delta(t))$ constant for all $t$ if the distribution of $\lambda_t$  is stable over time. However, the idea that treatment assignment is correlated with the factor model structure essentially means that the distribution of $\lambda_t$  is different before and after the treatment assignment. In this case,  it would not be reasonable to assume that the parameters of the linear projection of $y_{0t}^N$ given $\mathbf{y}_t$ are the same for $t \in \mathcal{T}_0$ and $t \in \mathcal{T}_1$ if we consider that  treatment assignment is correlated with the factor model structure. \cite{2018arXiv181210820C} assume that  $y_{0t}^N$ and $\mathbf{y}_t$  are covariance-stationary for all periods (see their Assumption 6), which implies that $(\delta_1(t),\delta(t))$ constant for all $t$. Therefore, they  also implicitly imply that there is no selection on unobservables. Since they consider a setting with   both large $J$ and $T$, however, it is possible that their estimator is consistent when there is selection on unobservables under conditions similar to the ones considered by \cite{Ferman}.  

\cite{Carvalho2015},  \cite{Carvalho2016b},  \cite{Masini}, and \cite{Zhou}  assume that the outcome of the control units are independent from treatment assignment. If we consider the linear factor model structure from Assumption \ref{assumption_LFM}, then this essentially means that there is no selection on unobservables. Given Assumption \ref{assumption_exogeneity}, if treatment assignment is correlated with the potential outcomes of the treated unit, then it must be correlated with $\lambda_t \mu_0$. However, if this is the case, then treatment assignment must also be correlated with at least some control units, implying that their assumption that  the outcome of the control units are independent from treatment assignment would be violated. Note that \cite{Carvalho2015},   \cite{Masini}, and \cite{Zhou}  encompass a setting with  both large $J$ and $T$. Therefore, it might be possible to consider a different set of assumptions, as the ones considered by \cite{Ferman}, so that their estimator is asymptotically unbiased when $J$ also increases.

Overall, our results clarify what selection on unobservables means in this setting, and the  conditions under which these estimators are asymptotically unbiased when $J$ is fixed. These results also clarify that there is no contradiction between these papers and the literature on factor models, which shows that factor loadings can only be consistently estimated with fixed $J$ under strong assumptions on the idiosyncratic shocks.

\pagebreak

\section{Appendix Tables and Figures}

\begin{table}[H]
  \centering
\caption{{\bf MC Results - Specification Test}} \label{Table_spec}
      \begin{lrbox}{\tablebox}
      \begin{adjustbox}{width=\textwidth,totalheight=0.8\textheight}
\begin{tabular}{cccccccccc}
\hline
\hline

&  \multicolumn{4}{c}{ No break }  & &  \multicolumn{4}{c}{Break in $\lambda_{1t}$  }     \\ \cline{2-5} \cline{7-10}

 & $T_0=120$ & $T_0=240$ & $T_0=480$ & $T_0=1200$ &  & $T_0=120$ & $T_0=240$ & $T_0=480$ & $T_0=1200$ \\

$\mu_{10}$     & (1) &  (2) & (3) &(4)  &      & (5) &  (6) & (7) & (8)  \\ 
     
     \hline
     
-2.6689 & 0.172 & 0.110 & 0.079 & 0.059 &  & -2.669 & 0.619 & 0.510 & 0.430 \\
-2.4079 & 0.179 & 0.117 & 0.084 & 0.063 &  & -2.408 & 0.691 & 0.582 & 0.503 \\
-1.5034 & 0.183 & 0.121 & 0.090 & 0.063 &  & -1.503 & 0.670 & 0.562 & 0.477 \\
-1.4303 & 0.170 & 0.115 & 0.091 & 0.064 &  & -1.430 & 0.594 & 0.477 & 0.391 \\
-1.1359 & 0.167 & 0.112 & 0.090 & 0.063 &  & -1.136 & 0.560 & 0.429 & 0.350 \\
-1.0772 & 0.168 & 0.117 & 0.091 & 0.068 &  & -1.077 & 0.666 & 0.553 & 0.484 \\
-1.0604 & 0.173 & 0.111 & 0.090 & 0.069 &  & -1.060 & 0.626 & 0.526 & 0.448 \\
-1.0173 & 0.165 & 0.111 & 0.085 & 0.061 &  & -1.017 & 0.598 & 0.507 & 0.430 \\
-1.0066 & 0.167 & 0.117 & 0.088 & 0.059 &  & -1.007 & 0.576 & 0.481 & 0.385 \\
-0.8201 & 0.150 & 0.114 & 0.080 & 0.065 &  & -0.820 & 0.616 & 0.563 & 0.506 \\
-0.8087 & 0.151 & 0.110 & 0.080 & 0.061 &  & -0.809 & 0.577 & 0.489 & 0.431 \\
-0.6899 & 0.170 & 0.125 & 0.095 & 0.064 &  & -0.690 & 0.412 & 0.313 & 0.218 \\
-0.6813 & 0.145 & 0.105 & 0.081 & 0.068 &  & -0.681 & 0.534 & 0.476 & 0.447 \\
-0.6594 & 0.158 & 0.116 & 0.098 & 0.061 &  & -0.659 & 0.459 & 0.362 & 0.315 \\
-0.6573 & 0.152 & 0.120 & 0.097 & 0.060 &  & -0.657 & 0.479 & 0.382 & 0.298 \\
-0.5299 & 0.155 & 0.109 & 0.085 & 0.063 &  & -0.530 & 0.374 & 0.287 & 0.229 \\
-0.4925 & 0.138 & 0.098 & 0.074 & 0.059 &  & -0.493 & 0.412 & 0.345 & 0.326 \\
-0.3721 & 0.156 & 0.113 & 0.092 & 0.063 &  & -0.372 & 0.324 & 0.244 & 0.187 \\
-0.3253 & 0.158 & 0.128 & 0.103 & 0.065 &  & -0.325 & 0.291 & 0.223 & 0.163 \\
-0.2952 & 0.126 & 0.101 & 0.088 & 0.060 &  & -0.295 & 0.321 & 0.265 & 0.230 \\
-0.1566 & 0.138 & 0.080 & 0.070 & 0.049 &  & -0.157 & 0.270 & 0.183 & 0.144 \\
-0.1291 & 0.136 & 0.116 & 0.086 & 0.060 &  & -0.129 & 0.214 & 0.167 & 0.120 \\
-0.1251 & 0.138 & 0.115 & 0.107 & 0.066 &  & -0.125 & 0.233 & 0.178 & 0.141 \\
-0.1190 & 0.153 & 0.121 & 0.097 & 0.062 &  & -0.119 & 0.271 & 0.192 & 0.133 \\
-0.1147 & 0.136 & 0.100 & 0.074 & 0.062 &  & -0.115 & 0.243 & 0.170 & 0.121 \\
-0.0297 & 0.145 & 0.120 & 0.103 & 0.066 &  & -0.030 & 0.225 & 0.163 & 0.119 \\
-0.0155 & 0.131 & 0.100 & 0.073 & 0.057 &  & -0.015 & 0.202 & 0.139 & 0.098 \\
0.1411 & 0.129 & 0.112 & 0.089 & 0.063 &  & 0.141 & 0.258 & 0.184 & 0.130 \\
0.1616 & 0.126 & 0.105 & 0.087 & 0.059 &  & 0.162 & 0.261 & 0.202 & 0.160 \\
0.1895 & 0.150 & 0.116 & 0.093 & 0.063 &  & 0.190 & 0.247 & 0.178 & 0.133 \\
0.2039 & 0.152 & 0.125 & 0.104 & 0.066 &  & 0.204 & 0.233 & 0.169 & 0.127 \\
0.2043 & 0.145 & 0.115 & 0.086 & 0.059 &  & 0.204 & 0.248 & 0.181 & 0.113 \\
0.3557 & 0.135 & 0.115 & 0.100 & 0.064 &  & 0.356 & 0.408 & 0.359 & 0.288 \\
0.3874 & 0.152 & 0.106 & 0.076 & 0.058 &  & 0.387 & 0.350 & 0.274 & 0.201 \\
0.5107 & 0.152 & 0.102 & 0.081 & 0.057 &  & 0.511 & 0.383 & 0.297 & 0.248 \\
0.6244 & 0.157 & 0.112 & 0.093 & 0.058 &  & 0.624 & 0.512 & 0.419 & 0.337 \\
0.6743 & 0.153 & 0.120 & 0.096 & 0.057 &  & 0.674 & 0.536 & 0.439 & 0.345 \\
0.6887 & 0.155 & 0.102 & 0.083 & 0.056 &  & 0.689 & 0.466 & 0.355 & 0.307 \\
0.7582 & 0.148 & 0.105 & 0.080 & 0.067 &  & 0.758 & 0.504 & 0.421 & 0.381 \\
0.7728 & 0.161 & 0.110 & 0.093 & 0.058 &  & 0.773 & 0.461 & 0.356 & 0.284 \\
0.9193 & 0.160 & 0.108 & 0.082 & 0.067 &  & 0.919 & 0.593 & 0.486 & 0.429 \\
0.9395 & 0.157 & 0.111 & 0.086 & 0.061 &  & 0.939 & 0.650 & 0.583 & 0.522 \\
0.9810 & 0.182 & 0.111 & 0.080 & 0.061 &  & 0.981 & 0.621 & 0.514 & 0.451 \\
1.1221 & 0.159 & 0.112 & 0.093 & 0.068 &  & 1.122 & 0.594 & 0.497 & 0.421 \\
1.2940 & 0.173 & 0.117 & 0.092 & 0.056 &  & 1.294 & 0.629 & 0.527 & 0.450 \\
1.3090 & 0.186 & 0.126 & 0.083 & 0.064 &  & 1.309 & 0.687 & 0.578 & 0.506 \\
1.3762 & 0.187 & 0.128 & 0.095 & 0.063 &  & 1.376 & 0.719 & 0.609 & 0.519 \\
1.3897 & 0.176 & 0.108 & 0.086 & 0.068 &  & 1.390 & 0.659 & 0.546 & 0.467 \\
1.5060 & 0.168 & 0.119 & 0.084 & 0.068 &  & 1.506 & 0.601 & 0.494 & 0.413 \\
1.6281 & 0.178 & 0.120 & 0.087 & 0.060 &  & 1.628 & 0.692 & 0.586 & 0.498 \\
2.1912 & 0.189 & 0.119 & 0.086 & 0.065 &  & 2.191 & 0.712 & 0.598 & 0.513 \\

     \hline

\hline

\end{tabular}
\end{adjustbox}
   \end{lrbox}
\usebox{\tablebox}\\
\settowidth{\tableboxwidth}{\usebox{\tablebox}} \parbox{\tableboxwidth}{\footnotesize{Notes: this table presents rejection rates for the specification test presented in Section \ref{Sub_demeaned}. In columns 1 to 4, there is no structural break, while in columns 5 to 8 the first common factor has expected value equal to two times its standard deviation in the post-treatment periods. 
}
}
\end{table}

\begin{table}[ht]
\centering
\caption{Estimated weights - Empirical Illustration}  \label{Appendix_table_weights}
\begin{tabular}{cccc}
  \hline
 & Original SC & Demeaned SC & Abadie et al. (2003) \\ 
  \hline
Andalucia & 0.0000 & 0.0000 & 0.0000 \\ 
  Aragon & 0.0000 & 0.0000 & 0.0000 \\ 
  Baleares (Islas) & 0.3111 & 0.2539 & 0.0000 \\ 
  Canarias & 0.0000 & 0.0000 & 0.0000 \\ 
  Cantabria & 0.0000 & 0.0008 & 0.0000 \\ 
  Castilla Y Leon & 0.0000 & 0.0002 & 0.0000 \\ 
  Castilla-La Mancha & 0.0000 & 0.0000 & 0.0000 \\ 
  Cataluna & 0.0000 & 0.0536 & 0.8508 \\ 
  Comunidad Valenciana & 0.0000 & 0.0003 & 0.0000 \\ 
  Extremadura & 0.0000 & 0.0000 & 0.0000 \\ 
  Galicia & 0.0000 & 0.0000 & 0.0000 \\ 
  Madrid (Comunidad De) & 0.4831 & 0.2879 & 0.1492 \\ 
  Murcia (Region de) & 0.0000 & 0.1898 & 0.0000 \\ 
  Navarra & 0.0000 & 0.0190 & 0.0000 \\ 
  Principado De Asturias & 0.0000 & 0.0072 & 0.0000 \\ 
  Rioja (La) & 0.2058 & 0.1873 & 0.0000 \\ 
   \hline
\end{tabular}
\end{table}

\pagebreak

\newpage

\singlespacing
\bibliographystyle{aer}
\bibliography{bib/bib.bib}

\end{document}